\documentclass[12pt]{article}

\usepackage{amsmath}
\usepackage[round]{natbib}
\usepackage{amssymb,color}
\usepackage{appendix}
\usepackage{amsthm}
\usepackage{enumerate}
\usepackage{hyperref}
\usepackage{mathtools} 
\mathtoolsset{showonlyrefs=true}  

\usepackage{caption}

\usepackage[margin=0.8in]{geometry}

\numberwithin{equation}{section}
\usepackage{comment}

\usepackage{titling}
\newtheorem{thm}{Theorem}[section]

\newtheorem{prop}[thm]{Proposition}

\newtheorem{lemma}[thm]{Lemma}

\newtheorem{remark}{Remark}[section]

\begin{document}

\author{Hyungbin Park\thanks{Department of Mathematical Sciences and RIMS, Seoul National University, 1, Gwanak-ro, Gwanak-gu, Seoul, Republic of Korea; hyungbin@snu.ac.kr,  hyungbin2015@gmail.com}  \\ \\
}
\title{Influence of risk tolerance on long-term investments: A Malliavin calculus approach}

\maketitle

\abstract{
This study investigates the influence of risk tolerance on the expected utility in the long run.
We estimate the extent to which the 
expected utility of optimal portfolios is affected by small changes in the risk tolerance.
For this purpose, we adopt the Malliavin calculus method and the Hansen--Scheinkman decomposition, through which the expected utility is expressed in terms of 
the eigenvalues and eigenfunctions
of an operator.
We conclude that the influence of risk aversion on the expected utility is determined by
these eigenvalues and eigenfunctions in the long run. 

}

\section{Introduction}

\subsection{Overview}

The risk--return trade-off is an important issue in finance.
Investors assess returns against risk while considering investment strategies, and they choose strategies that maximize the returns while minimizing the risk.
There are several ways to formulate the risk--return trade-off. One of the most commonly accepted forms is the utility function, which reflects the risk tolerance of an investor. 
We quantify the extent to which the 
expected utility is affected by small changes in the risk tolerance.

This study investigates long-term investment strategies.
Under several market models, utility-maximizing portfolios and their expected utility are considered.
For this purpose, we adopt the Malliavin calculus method and the Hansen--Scheinkman decomposition, through which the expected utility is expressed in terms of 
the eigenvalues and eigenfunctions
of an operator.
We conclude that the influence of risk tolerance on the long-term expected utility is determined by these eigenvalues and eigenfunctions. \newline

The main objective of this study is to investigate the influence of small changes 
of the utility function on long-term investments.
We recall the classical utility maximization problem
$$-u_T:=\max\mathbb{E}[U(\Pi_T)]$$
over all admissible self-financing portfolios
with a given initial capital for terminal time $T.$
The constant relative risk aversion (CRRA) utility function 
$$U(x)=\frac{x^{\nu}}{\nu}$$
for $\nu<0$ is considered in this paper. The parameter $\nu$ represents how an investor measures his/her degree of risk tolerance.
The optimal portfolio $\hat{\Pi}=\hat{\Pi}^{(\nu)}$ depends on the parameter $\nu,$ thus
$$-u_T=\mathbb{E}[U(\hat{\Pi}_T)]=\mathbb{E}[\hat{\Pi}_T^\nu]/\nu\,.$$
We investigate the influence of small changes in the parameter $\nu$ on the expected utility over a long time horizon. 
The influence of small changes in the parameter $\nu$ in terms of its logarithmic value can be mathematically expressed as
$$\frac{\partial }{\partial \nu} \ln u_T\,.$$
We estimate the large-time behavior of this partial derivative as $T\to\infty.$

The main methodology for this analysis is a combination of the Hansen--Scheinkman decomposition and the Malliavin calculus technique presented in Sections \ref{sec:math_prel} and \ref{sec:main_argu}. This approach is from \cite{park2018sensitivity}. First, we transform the expected utility into the expectation form 
$$p_T=\mathbb{E}_\xi[e^{-\int_0^Tq(X_s)\,ds}]$$
for some Markov diffusion process $X=(X_t)_{0\le t\le T}$ with $X_0=\xi$ and some measurable function $q.$
The process $X$ with killing rate $q$ induces an infinitesimal generator.
Using the Hansen--Scheinkman decomposition, 
we can find an eigenvalue $\lambda$ and a positive eigenfunction $\phi$ of the generator as well as a measurable function $f$ such that
the expectation is written as
$$p_T=\phi(\xi)e^{-\lambda T}
f(T,\xi)\,.$$ 
The real number $\lambda$ and the functions $\phi$ and $f$ depend on the parameter $\nu.$
By differentiating with respect to $\nu,$ we have
$$\frac{\partial}{\partial\nu}\ln p_T=\frac{\partial}{\partial\nu}\ln\phi(\xi)-T\frac{\partial\lambda}{\partial\nu}+\frac{\partial}{\partial\nu}\ln f(T,\xi)\,.$$

Using the above-mentioned equation, we estimate the influence of risk tolerance on the long-term expected utility.
If $\frac{\partial}{\partial\nu}\ln f(T,\xi)$
is bounded in $T$ on $[0,\infty),$ then 
$$\Big|\frac{1}{T}\frac{\partial}{\partial\nu}\ln p_T+\frac{\partial\lambda}{\partial\nu}\Big|\leq \frac{c}{T}$$
for some positive constant $c.$
This implies that the influence of risk tolerance
is asymptotically equal to the partial derivative of $-\lambda$ with respect to $\nu,$ which is the main conclusion of this study.
The Malliavin calculus technique is used to verify that $\frac{1}{T}\frac{\partial}{\partial\nu}\ln f(T,\xi)$
is bounded in $T$ on $[0,\infty).$
We cover several market models, namely the Ornstein--Uhlenbeck process, the CIR process, the $3/2$ model, a quadratic drift model.


The remainder of this paper is organized as follows.  
The related literature is reviewed in Section \ref{sec:related}.
The utility maximization problem, the Hansen--Scheinkman decomposition and the Malliavin calculus method are explained as mathematical preliminaries in Section \ref{sec:math_prel}.  
The main ideas and arguments 
for investigating the influence of risk tolerance are discussed in Section \ref{sec:main_argu}.
The influence of risk tolerance
on utility-maximizing portfolios is illustrated in Section \ref{sec:u_max}.
Finally, our findings are summarized in Section \ref{sec:con}.
The technical details are presented in the appendices.

\subsection{Related literature}
\label{sec:related}

Many authors have studied the stability
of the optimal investment strategy
with respect to the utility function.
\cite{jouini2004convergence} 
studied in a general complete financial market the stability
of the optimal investment-consumption strategy
with respect to the choice of the utility function. More precisely, for a given
sequence of utility functions that converges pointwise, they proved the almost sure as well
as the $L^p$-convergence $(p\ge1)$ of the optimal wealth and consumption at each date. 
\cite{carassus2007optimal}
investigated the convergence of optimal strategies 
with respect to a sequence of utility functions.
They also established the continuity of the utility indifference price with respect to changes in agents' preferences.
\cite{nutz2012risk}
considered the economic problem of optimal consumption and
investment with power utility. 
As the relative risk aversion
tends to infinity or to one was proved,
the convergence of the optimal consumption is obtained for
general semimartingale models while the convergence of the optimal trading strategy
is obtained for continuous models.

The dependence of the risk tolerance on the investment strategy has been studied many authors.
\cite{zariphopoulou2009investment} analyzed a portfolio choice problem when the local risk tolerance is time-dependent and asymptotically linear in wealth.
This methodology allows the investment performance to be measured in terms of the risk tolerance and alternative market views.
\cite{mocha2013stability}
studied the sensitivity of the power utility maximization problem with respect to the
investor's relative risk aversion, the statistical probability measure, the investment constraints, and
the market price of risk.
\cite{paravisini2017risk}
estimated risk tolerance from investors' financial decisions in a person-to-person lending platform. They developed a method that obtains a risk-tolerance parameter from each portfolio choice on the basis of the elasticity of risk tolerance to changes in wealth. 
\cite{bi2019optimal}
investigated the optimal investment--reinsurance strategies for an agent with state-dependent risk tolerance and value-at-risk constraints. 
They derived the closed-form expressions of the optimal strategies and discussed the impact of the risk tolerance.
\cite{delong2019optimal}
considered agents whose risk tolerance consists of a constant risk tolerance and a small wealth-dependent risk tolerance. He
investigated an exponential utility maximization problem for an agent who faces a stream of non-hedgeable claims.

Numerous studies have investigated the topic of long-term investment strategies.
\cite{fleming2000risk} considered an optimal investment model  to maximize the long-term growth rate of the expected utility of wealth. 
The problem was reformulated as a risk‐sensitive control problem with an infinite time horizon. 
\cite{hansen2009long}, \cite{hansen2012dynamic}   and \cite{hansen2012pricing}  
exploited the Hansen--Scheinkman decomposition method and 
demonstrated
a long-term risk--return trade-off.
\cite{guasoni2015static} studied a class of static fund separation theorems that is valid for investors with a long time horizon and constant relative risk tolerance.
\cite{robertson2015large} investigated long-term portfolio choice problems by analyzing the large-time asymptotic behavior of solutions to semi-linear Cauchy problems.

Malliavin calculus has been studied in relation to various topics in quantitative finance. 
\cite{fournie1999applications}  
investigated a probabilistic method
for computations of Greeks in finance.
This methodology is based on the integration-by-parts formula developed in Malliavin calculus. 
\cite{benhamou2003optimal}
showed that the Malliavin weight functions for Greeks must satisfy necessary and sufficient conditions expressed as conditional expectations.
\cite{alos2007short} used Malliavin calculus techniques to obtain an expression for the short-time behavior of the at-the-money implied volatility skew for a generalization of the Bates model, where the volatility does not need to be a diffusion or
a Markov process.
\cite{borovivcka2014shock}
studied the shock elasticity, which reflects the sensitivity with respect
to a perturbation over a time instant. 
They proposed a Malliavin calculus method to compute the shock elasticity. 
\cite{alos2019estimating} studied the difference between the fair strike of a volatility swap and the at-the-money implied volatility of a European call option.
They used the Malliavin calculus approach to derive an exact expression for this difference. 
\cite{park2019sensitivity} employed the Malliavin calculus method to investigate
the sensitivities of the long-term expected utility of optimal portfolios under incomplete markets.

As a closely related article,  \cite{park2018sensitivity} 
conducted a sensitivity analysis of long-term cash flows.
However, the perturbation form is different from this paper.  He studied the extent to which the price of the cash flow is affected by small perturbations of the underlying Markov diffusion.
He considered the drift and volatility perturbations in the underlying process
$$dX_{t}^{\epsilon}=b_\epsilon(X_{t}^{\epsilon})\,dt+\sigma_\epsilon (X_{t}^{\epsilon})\,dB_{t}\;,\,X_0^\epsilon=\xi$$
for the perturbation parameter $\epsilon$, and 
analyzed their influence to the cash flows. This paper, however, works with the investor's risk tolerance and the perturbations in the underlying process is not considered.

\section{Mathematical preliminary}
\label{sec:math_prel}

In this section, we demonstrate the utility maximization problem and the two main methodologies exploited in the rest of this study: 
the Hansen--Scheinkman decomposition and the 
Malliavin calculus method.

\subsection{Utility maximization problem}
\label{sec:u_max_prob}

We consider the classical utility maximization problem in  financial markets. Let $(\Omega,\mathcal{F},(\mathcal{F}_t)_{t\ge0},{\bf P})$ be a filtered probability space having a one-dimensional Brownian motion $(Z_t)_{t\ge0}$ and a $n$-dimensional Brownian motion  $W$ with constant correlation $\rho=(\rho_1,\cdots,\rho_n)^\top=d\left\langle Z, \, W \right\rangle_t / dt.$
The measure ${\bf P}$ is referred to as the  physical measure and the filtration $(\mathcal{F}_t)_{t\ge0}$ is the argumentation of the natural filtration of $Z$ and $W.$ 
We assume that the market has a state process $(X_t)_{t\ge0},$ which is a Markov diffusion satisfying
\begin{equation}
\label{eqn:X}
dX_t=b(X_t)\,dt+\sigma(X_t)\,dB_t\,,\;X_0=\xi
\end{equation}
There are  $n+1$ assets $S^{(0)},S^{(1)},\ldots,S^{(n)}$ in the market, where $S^{(0)}$ is a risk-free asset and $S^{(1)},\ldots,S^{(n)}$ are risky assets. The assets dynamics are given as
\begin{equation}   \label{eqn:S}
\begin{aligned}
&\frac{dS^{(0)}_t}{S^{(0)}_t} = r(X_t)\,  dt\,,\\
&\dfrac{dS^{(i)}_t}{S^{(i)}_t} = (r(X_t)+  \mu_i(X_t)) dt + \sum_{j=1}^{n} \upsilon_{ij}(Y_t) \, dW^{(j)}_t \qquad 1 \leq i \leq n\,.
\end{aligned}
\end{equation}
Here, $r$ is the short interest rate function, $\mu=(\mu_1,\cdots,\mu_n)^\top$ is the excess return function, and $\upsilon=(\upsilon_{ij})_{1\le i,j\le n}$ is the volatility matrix functions.

A portfolio is a $n$-dimensional process $\pi=(\pi_t)_{t\ge0}=(\pi^{(1)}_t,\ldots,\pi^{(n)}_t)_{t\geq0}$, which represents the proportions of wealth in each risky asset.
The wealth process $(\Pi_t)_{t\ge0}=(\Pi_t^\pi)_{t\ge0}$ of $\pi$ with positive initial capital $\omega$ satisfies
\begin{equation} \label{X_SDE}
\frac{d\Pi_t}{\Pi_t} =(r(X_t)+\pi_t\mu(X_t)) \, dt + \pi_t \upsilon(X_t) \, dW_t, \quad \Pi_0 =\omega>0.
\end{equation}
We assume that the portfolio process $(\pi_t)_{t\geq0}$ is  $(\mathcal{F}_t)_{t\ge0}$-adapted and the integrations in Eq.\eqref{X_SDE} are well-defined, that is, $\mathbb{E}\int_0^t |(r(X_s)+\pi_s\mu(X_s))|+|\!|\pi_s\upsilon(X_s) |\!|^2\,ds<\infty$ for all $t\ge0,$ where $|\!|\cdot|\!|$ is the usual Euclidean norm.
The wealth process $\Pi_t>0$ for all $t\ge0$ a.s. since the initial wealth $\omega>0.$

In this market, we consider an agent who wants to maximize the expected utility of the terminal wealth over all possible portfolios.
More precisely, we are interested in 
$$\sup_\pi\mathbb{E}^{\bf P}[U(\Pi_T^{\pi})]$$
over all possible portfolios $\pi$ for given positive initial capital $\omega.$
The utility function is assumed to be a power function of the form  
\begin{equation}
\label{eqn:utility_ftn}
U(x)=x^\nu/\nu
\end{equation} 
for $\nu<0.$

This utility maximization problem can be solved by using stochastic control theory.
Define the value function $u_T$ as
$$-u_T=-u_T(t,\omega,x)=\sup_{\pi}\mathbb{E}^{\bf P}[U(\Pi_T^{\pi})|\Pi_t=\omega,X_t=x]\,.$$ 
Following \cite[Proposition 2.1]{zariphopoulou2001solution}, we have
$$u_T(t,\omega,x)=-\frac{\omega^\nu}{\nu} (\mathbb{E}^\mathbb{P}[e^{\frac{\nu(1-\nu+\nu\rho'\rho)}{1-\nu}\int_t^T (r+\frac{1}{2(1-\nu)}\theta'\theta)(X_u)\,du}|X_t=x])^{\frac{1-\nu}{1-\nu+\nu \rho'\rho}}$$
where $\theta:=\upsilon^{-1}\mu,$ $b:=k+\frac{\nu}{1-\nu}\sigma\rho'\theta$ and the $\mathbb{P}$-dynamics of $X$ is
\begin{equation}
\label{eqn:X_P}
dX_t=b(X_t)\,dt+\sigma(X_t)\,dB_t
\end{equation} 
with a $\mathbb{P}$-Brownian motion $(B_t)_{t\ge0}.$
For $t=0,$ it follows that
$$u_T=u_T(0,\omega,\xi)=-\frac{\omega^\nu}{\nu} q_T^{\frac{1-\nu}{1-\nu+\nu \rho'\rho}}$$
for
\begin{equation}\label{eqn:q_T}
q_T:=\mathbb{E}^\mathbb{P}_\xi[e^{\frac{\nu(1-\nu+\nu\rho'\rho)}{1-\nu}\int_0^T (r+\frac{1}{2(1-\nu)}\theta'\theta)(X_u)\,du}]\,.
\end{equation}
Here, we have used the notation $\mathbb{E}^\mathbb{P}_\xi[\,\cdots]=\mathbb{E}^\mathbb{P}[\,\cdots|X_0=\xi].$
In particular, when the short rate is a constant $r,$
\begin{equation}
\label{eqn:-u_T}
u_T=-\frac{\omega^\nu}{\nu} e^{r\nu T} p_T^{\frac{1-\nu}{1-\nu+\nu \rho'\rho}}
\end{equation}
where
\begin{equation}
\label{eqn:p_T}
p_T:=\mathbb{E}_\xi^{{\mathbb{P}}}[e^{-q\int_t^T (\theta'\theta)(X_u)\,du}]
\end{equation}
and $q:=-\frac{\nu(1-\nu+\nu\rho'\rho)}{2(1-\nu)^2}.$
 
In conclusion, the problem of analyzing the optimal expected utility boils down to  the problem of analyzing the expectation $q_T$ in Eq.\eqref{eqn:q_T} (or $p_T$ in Eq.\eqref{eqn:p_T} if the short rate is constant) with the dynamics of $X$ given as Eq.\eqref{eqn:X_P}.
Later this will be used in {\bf Step I} in Section \ref{sec:main_argu}.

\subsection{Hansen--Scheinkman decomposition}
\label{sec:HS}

We briefly review the 
Hansen--Scheinkman decomposition as a mathematical preliminary.
Readers may refer to \cite{hansen2009long}, \cite{park2018sensitivity}, and \cite{qin2016positive} for further details. 
Consider a filtered probability space $(\Omega,\mathcal{F},(\mathcal{F}_t)_{t\ge0},{\mathbb{P}})$ having a $d$-dimensional Brownian motion $B=(B_t^{(1)},\cdots,B_t^{(d)})_{t\ge0}^\top.$
The family $(\mathcal{F}_t)_{t\ge0}$ is the completed filtration generated by $B.$

We begin with the state space $\mathcal{D}$ and four functions $b,\sigma,q,h.$
Let $\mathcal{D}$ be an open and connected subset of $\mathbb{R}^d,$ and let $b:\mathcal{D}\to\mathbb{R}^d$ and $\sigma:\mathcal{D}\to\mathbb{R}^d\times\mathbb{R}^d$  be continuously differentiable functions. The matrix $\sigma$ is invertible. 	The function $q:\mathcal{D}\to\mathbb{R}$ is continuous and the function $h:\mathcal{D}\to\mathbb{R}$ is nonnegative continuously differentiable.	For each $\xi\in\mathcal{D},$ assume that the SDE
Eq.\eqref{eqn:X} has a unique strong solution $X$ on $\mathcal{D}.$

We can 
consider an infinitesimal generator and its eigenpair.
Define the infinitesimal generator of the process $X$ with killing rate $q$ as
\begin{equation} \label{eqn:L}
\mathcal{L}:=\frac{1}{2}\sum_{i,j=1}^da_{ij}\frac{\partial^2}{\partial x_i\partial x_j}+\sum_{i=1}^db_i\frac{\partial}{\partial x_i}-q,
\end{equation}
where $a=\sigma\sigma^\top.$
For a real number $\lambda$ and a positive $C^2$-function  $\phi:\mathcal{D}\to\mathbb{R},$ we say that
a pair $(\lambda,\phi)$
is an {\em eigenpair} of $-\mathcal{L}$ if
\begin{equation}\label{eqn:eigen_semigp}
\mathcal{L}\phi=-\lambda\phi \quad\textnormal{on } \mathcal{D}\,. 
\end{equation}
For $T>0,$ by applying the Ito formula, we can show that for each eigenpair $(\lambda,\phi),$ a positive process
\begin{equation}\label{eqn:M}
M_t^\phi:=\frac{\phi(X_{t})}{\phi(\xi)}e^{\lambda t-\int_{0}^{t} q(X_{s})\,ds}\,,\;0\le t\le T 
\end{equation}
is a local martingale.
Assume that this process is a martingale. 
We can define a measure $\hat{{\mathbb{P}}}_t^\phi$ on each
$\mathcal{F}_t$ for $0\le t\le T$
by
$$\frac{d\hat{{\mathbb{P}}}_t^\phi}{d{\mathbb{P}}\;}=M_t^\phi\,.$$
The family $(\hat{{\mathbb{P}}}_t^\phi)_{0\le t\le T}$  is consistent, i.e., 
$\hat{{\mathbb{P}}}_{t'}^\phi|_{\mathcal{F}_t}=\hat{{\mathbb{P}}}_t^\phi$  
for all $0\le t\le t'\le T.$ 
For fixed $T>0,$ we use the notation $\hat{\mathbb{P}}^\phi$ instead of $\hat{{\mathbb{P}}}_T^\phi,$ suppressing $T.$ This measure on $\mathcal{F}_T$ is called the eigen-measure with respect
to  $\phi.$
The process
$$\hat{B}_t^\phi=-\int_0^t(\sigma^\top\nabla\phi/\phi)(X_s)\,ds+B_t\,,\;t\ge0$$
is a $\hat{{\mathbb{P}}}^\phi$-Brownian motion
by the Girsanov theorem, and  
the process
$X$ satisfies
$$dX_t=(b+\sigma\sigma^\top\nabla\phi/\phi)(X_t)\,dt+\sigma(X_t)\,d\hat{B}_t^\phi\,. $$

Under these circumstances, consider the decomposition of the discount factor
\begin{equation} 
e^{-\int_{0}^{t} q(X_{s})\,ds}=M_t^\phi e^{-\lambda t}\frac{\phi(\xi)}{\phi(X_{t})}\,,\;t\ge0, 
\end{equation} 
which comes from Eq.\eqref{eqn:M}. 
This expression is called the {\em Hansen--Scheinkman decomposition}. 
The expectation $p_T$ can be written as
\begin{equation}  \label{eqn:HS_tranform}
\begin{aligned}
p_T=\mathbb{E}_\xi^{\mathbb{P}}[e^{-\int_0^T q(X_s)\,ds}]
&=\phi(\xi)\,e^{-\lambda T}\,\mathbb{E}_\xi^{\mathbb{P}}
\Big[\frac{M_T^\phi}{\phi(X_T)}\Big]\\
&=\phi(\xi)\,e^{-\lambda T}\,
\mathbb{E}_\xi^{\hat{\mathbb{P}}^\phi }
\Big[\frac{1}{\phi(X_T)}\Big]\\
&=\phi(\xi)e^{-\lambda T}
f_\phi (T,\xi),
\end{aligned}
\end{equation}
where $f_\phi (t,x):=\mathbb{E}_x^{{\hat{\mathbb{P}}}^\phi}
[\frac{1}{\phi(X_t)}]$ for $t\ge0$ and $x\in\mathcal{D}.$
The function $f_\phi $ is referred to as the {\em remainder function}.  
The decomposition in Eq.\eqref{eqn:HS_tranform} is useful for the analysis of $p_T,$
because the expectation 
$\mathbb{E}_\xi^{{\hat{\mathbb{P}}}^\phi}
[\frac{1}{\phi(X_T)}]$
depends on
the final random variable $X_T,$ whereas 
the expression $$p_T=\mathbb{E}_\xi^{\mathbb{P}}[e^{-\int_0^T q(X_s)\,ds}]$$ depends on the entire path of
$(X_t)_{0\le t\le T}.$ If we know the $\hat{\mathbb{P}}^\phi$-distribution of $X_T$, then
we can analyze the expectation 
$\mathbb{E}_\xi^{\hat{\mathbb{P}}^\phi}
[\frac{1}{\phi(X_T)}]$
directly.
For notational simplicity, when the eigenpair is specified, we use 
the notations $M,$  $\hat{{\mathbb{P}}},$ $(\hat{B}_t)_{t\ge0}$, and $f$ instead of $M^\phi,$  $\hat{{\mathbb{P}}}^\phi,$ $(\hat{B}_t^\phi)_{t\ge0}$, and $f_\phi,$ respectively, suppressing $\phi.$

We summarize this section in the following proposition.
\begin{prop} Let $\mathcal{D}$ be an open and connected subset of $\mathbb{R}^d.$  	 Assume the following conditions.
	\begin{enumerate}[(i)]
		\item  The functions $b:\mathcal{D}\to\mathbb{R}^d$ and $\sigma:\mathcal{D}\to\mathbb{R}^d\times\mathbb{R}^d$  are continuously differentiable functions, and the matrix $\sigma$ is invertible. 	The function $q:\mathcal{D}\to\mathbb{R}$ is  continuous. 
		\item For each $\xi\in\mathcal{D},$  the SDE \eqref{eqn:X} has a unique strong solution $X$ on $\mathcal{D}.$ 	
	\end{enumerate}
If $(\lambda,\phi)$ is an {\em eigenpair} of the operator $-\mathcal{L}$ in Eq.\eqref{eqn:L}, then the process $M^\phi$ in Eq.\eqref{eqn:M} is a local martingale. If this process is a martingale, then
	$$	\mathbb{E}_\xi^{\mathbb{P}}[e^{-\int_0^T q(X_s)\,ds}]=\phi(\xi)\,e^{-\lambda T}\,
	\mathbb{E}_\xi^{\hat{\mathbb{P}}}
	\Big[\frac{1}{\phi(X_T)}\Big]$$
	where $\hat{\mathbb{P}}$ is the eigen-measure with respect to $\phi.$ 
\end{prop}

\subsection{Malliavin calculus}

This section presents a brief review of Malliavin calculus. 
For further details, refer to \cite{malliavin2006stochastic} and \cite{nualart2018introduction}.
Let $(\Omega,\mathcal{F},(\mathcal{F}_t)_{0\le t\le T},\mathbb{P})$ be a filtered probability space having a one-dimensional Brownian motion $B.$ The filtration $(\mathcal{F}_t)_{0\le t\le T}$ is the natural filtration of $B.$
Define the set of all cylindrical random variables as
$$\mathcal{S}:=\Big\{f\Big(\int_0^T h^{(1)}(s)\,dB_s,\cdots,\int_0^T h^{(n)}(s)\,dB_s\Big)  : n\in\mathbb{N}, h^{(1)},\cdots,h^{(n)}\in L^2[0,T], f\in C_p^\infty(\mathbb{R}^n)  \Big\}\,,$$
where $C_p^\infty(\mathbb{R}^n)$ is the set of all $C^\infty$ functions such that $f$ and all its
partial derivatives have polynomial growth.
The Malliavin derivative of $F\in\mathcal{S}$ is defined as
a stochastic process $DF=(D_tF)_{0\le t\le T}$  given by
$$D_tF=\sum_{i=1}^nh^{(i)}(t)\frac{\partial f}{\partial x_i}\Big(\int_0^T h^{(1)}(s)\,dB_s,\cdots,\int_0^T h^{(n)}(s)\,dB_s\Big)\,.   $$
Then, $D$ is a linear operator from $\mathcal{S}\subseteq L^2(\Omega)$ to $L^2([0,T]\times\Omega),$ and it is known that $D$ is closable. The closure is also denoted as $D.$
The domain of $D$ is the closure of $\mathcal{S}$ under the norm 
$$|\!|F|\!|_{\mathbb{D}^{1,2}}:=|\!|F|\!|_{L^{2}(\Omega)}+|\!|DF|\!|_{L^2([0,T]\times\Omega)}$$
and is denoted as $\mathbb{D}^{1,2}.$

We will use the following propositions.
Propositions \ref{prop:Malliavin_chain}, \ref{prop:Malliavin_first_variation}, and \ref{prop:rho} are from \cite[Lemma 2.1]{leon2003anticipating}, 
\cite[Theorem 39 on page 312]{protter2004stochastic},
and \cite[Proposition A.1]{park2018sensitivity}, respectively.

\begin{prop}\label{prop:Malliavin_chain}
	Let $\varphi:\mathbb{R}\to\mathbb{R}$ be a continuously differentiable function and $F\in\mathbb{D}^{1,2}.$ Then, $\varphi(F)\in \mathbb{D}^{1,2}$ if and only if $\varphi(F)\in L^2(\Omega)$ and $\varphi'(F)DF \in L^2([0,T]\times\Omega),$ and in this case,
	$$D\varphi(F)=\varphi'(F)DF\,.$$
\end{prop}

\begin{prop}\label{prop:Malliavin_first_variation} Let $X=X^{(x)}$ be a Markov diffusion whose dynamics is given as
	$$dX_t=b(X_t)\,dt+\sigma(X_t)\,dB_t$$
	with initial value $X_0=x$,
	where $b$ and $\sigma$ are  continuously differentiable functions with
	bounded derivatives. Then, the map $x\mapsto X_t^{(x)}$ is continuously differentiable almost surely and the derivative process
	$Y_t:=\frac{\partial }{\partial x}X_t^{(x)}$ satisfies
	$$dY_t=b'(X_t)Y_t\,dt+\sigma'(X_t)Y_t\,dB_t\,,Y_0=1\,,$$
	equivalently, $$Y_t=e^{\int_0^t(b'(X_s)-\frac{1}{2}\sigma'^2(X_s))\,ds+\int_0^t\sigma'(X_s)\,dB_s}\,.$$
	Moreover, $X_t\in\mathbb{D}^{1,2}$ for each $0\le t\le T$ and its Malliavin derivative satisfies
	$$D_sX_t=\sigma(X_s)\frac{Y_t}{Y_s}\mathbb{I}_{\{s\le t\}}.$$
\end{prop}

\begin{prop}\label{prop:IBP}
 (Integration by parts formula) Let $F\in\mathbb{D}^{1,2}$ and $(h_t)_{0\le t\le T}$ be a progressively measurable process with $\mathbb{E}\int_0^Th_s^2\,ds<\infty.$ Then
 $$\mathbb{E}\Big(F\int_0^Th_s\,dB_s\Big)=\mathbb{E}\Big(\int_0^T (D_sF)h_s\,ds\Big)\,.$$
\end{prop}

\begin{prop} \label{prop:rho}
	Let $b,$ $\overline{b},$ $\sigma$ be continuously differentiable  functions with $\sigma>0$ and let $f$ be a continuous function on an open interval $\mathcal{D}\subseteq\mathbb{R}.$ Define $b_\epsilon=b+\epsilon \overline{b}$ for $\epsilon \in I$, where $I$ is an open neighborhood of $0.$ Assume that the SDE 
	$$dX_t^{(\epsilon)}=b_\epsilon(X_t^{(\epsilon)}
	)\,dt+\sigma(X_t^{(\epsilon)})\,dB_t\,,\;X_0^{(\epsilon)}=\xi$$
	has a unique strong solution $X^{(\epsilon)}$ on $\mathcal{D}$ for all $\xi\in\mathbb{R}$ and $\epsilon\in I.$
	Suppose that for $T>0$, there exist positive constants $\epsilon_0,$ $\epsilon_1,$ $p,$ $q$ with $p\geq2$ and  $1/p+1/q=1$ such that 
	\begin{align}
	&\mathbb{E}_\xi[e^{\epsilon_0\int_0^T\overline{b}^2(X_s)\,ds}]<\infty\label{eqn:expo_condi}\,,\\
	&\mathbb{E}_\xi\Big[\int_0^T|\overline{b}|^{p+\epsilon_1}(X_s)\,ds\Big]<\infty\label{eqn:g_condi}\,,\\
	&\mathbb{E}_\xi[|f|^q(X_T)]<\infty\label{eqn:psi_condi}\,.
	\end{align}		
	Then, 
	the expectation	$\mathbb{E}[f(X_t^{(\epsilon)})]$ is continuously differentiable in $\epsilon$ and 
	\begin{equation}
	\begin{aligned}\label{eqn:deriva_rho}
	\frac{\partial}{\partial\epsilon}\Big|_{\epsilon=0}\mathbb{E}[f(X_T^\epsilon)]
	=\mathbb{E}\Big[f(X_{T})\int_{0}^{T}(\sigma^{-1}\overline{b})(X_{s})\, dB_{s}\Big],
	\end{aligned}
	\end{equation}
	where $X=X^{(0)}.$	
\end{prop}

\section{Main arguments}
\label{sec:main_argu}

This study investigates the influence of risk tolerance on the optimal expected utility in the long run. 
It involves the following steps. \newline

\noindent {\bf Step I}.   Transform the expected utility from the optimal investment strategy into the expectation form
$$p_T=\mathbb{E}^\mathbb{P}[e^{-\int_0^Tq(X_s)\,ds}],$$
where the $\mathbb{P}$-dynamics of $X$ is
$$dX_t=b(X_t)\,dt+\sigma(X_t)\,dB_t\,.$$
The drift function $b(\cdot)=b(\cdot;\nu)$  
and the killing rate $q(\cdot)=q(\cdot;\nu)$ may depend on $\nu;$
however, the volatility function $\sigma(\cdot)$ does not depend on $\nu.$ This step was conducted   in Section \ref{sec:u_max_prob}. \newline

\noindent {\bf Step II}.  
Through the Hansen--Scheinkman decomposition discussed in Section \ref{sec:HS},
the expectation can be expressed as
$$p_T=\phi(\xi)e^{-\lambda T}f(T,\xi)\,,$$
where $(\lambda,\phi)$ is an eigenpair and $f(T,\xi)=\mathbb{E}_\xi^{{\hat{\mathbb{P}}}}
[\frac{1}{\phi(X_T)}].$ The $\hat{\mathbb{P}}$-dynamics of $X$ is
$$dX_t=\kappa(X_t)\,dt+\sigma(X_t)\,d\hat{B}_t\,,$$
where $\kappa:=b+\sigma^2 \phi'/\phi.$
It follows that
\begin{equation} \label{eqn:main_argu_p_T}
\frac{\partial}{\partial \nu}\ln p_T=\frac{\partial}{\partial \nu}\ln\phi(\xi)-T\frac{\partial\lambda}{\partial \nu}
+\frac{f_\nu(T,\xi)}{f(T,\xi)}  \,.
\end{equation}

In the remainder function $f(T,\xi)=\mathbb{E}_\xi^{{\hat{\mathbb{P}}}}
[\frac{1}{\phi(X_T)}],$ observe that 
the drift function $\kappa(\cdot)=\kappa(\cdot;\nu),$  the eigenfunction $\phi(\cdot)=\phi(\cdot;\nu)$, and the measure $\hat{\mathbb{P}}=\hat{\mathbb{P}}^{\nu}$ depend on $\nu.$ For convenience, define $H(x;\nu)=\frac{1}{\phi(\cdot;\nu)}.$  
Then,
$$f_\nu(T,\xi)=\mathbb{E}_\xi^{\hat{\mathbb{P}}^\nu}
\Big[\frac{\partial}{\partial {\nu'}}\Big|_{\nu'=\nu}H(X_T;\nu')\Big]+\frac{\partial}{\partial \nu'}\Big|_{\nu'=\nu}\mathbb{E}_\xi^{\hat{\mathbb{P}}^{\nu'}}[H(X_T;\nu)]\,.$$
This is the key observation of this study. The perturbation of the risk aversion is transformed into perturbations of the drift function, payoff function, and eigenfunction. \newline

\noindent {\bf Step III}.
We prove (case by case for each model) that the term $\frac{f_\nu(T,\xi)}{f(T,\xi)}$ is bounded in $T$ on $[0,\infty).$ This is achieved as follows.
First, show that 
the denominator $f(T,\xi)$ converges to a positive constant as $T\to\infty.$
More precisely,
the process $X$ has an invariant distribution $\pi$ under the measure $\hat{\mathbb{P}}$ and 
the remainder function  
$$f(T,\xi)=\mathbb{E}_x^{{\hat{\mathbb{P}}}^\phi}
\Big[\frac{1}{\phi(X_t)}\Big] \to \int \frac{1}{\phi}\,d\pi$$ as $T\to\infty$ with  $\frac{1}{\phi}\in L^1(\pi).$
Second, show that
$$\mathbb{E}_\xi^{\hat{\mathbb{P}}^\nu}
\Big[\frac{\partial}{\partial {\nu'}}\Big|_{\nu'=\nu}H(X_T;\nu')\Big]$$ 
is bounded in $T$ on $[0.\infty).$
This can be easily checked by direct calculation; thus, we do not go into further detail here.

Finally, show that 
$$\frac{\partial}{\partial \nu'}\Big|_{\nu'=\nu}\mathbb{E}_\xi^{\hat{\mathbb{P}}^{\nu'}}[H(X_T;\nu)]$$
is bounded in $T$ on $[0,\infty).$  
Observe that the perturbation parameter $\nu$ is only in the drift term of the dynamics of $X.$ We adopt the Malliavin calculus method to estimate this partial derivative.
Assume that the map $\nu\mapsto \kappa(x;\nu)$ is continuously differentiable for each $x,$
and denote the first-order approximation as $\overline{\kappa}(x,\nu),$ i.e., 
$$\kappa(x;\nu+\epsilon)=\kappa(x;\nu)+\epsilon\overline{\kappa}(x;\nu)+o(\epsilon)$$
as $\epsilon\to0.$
Under suitable conditions (Propositions \ref{prop:Malliavin_chain}, \ref{prop:Malliavin_first_variation}, and \ref{prop:rho}),
\begin{equation}\label{eqn:main_Malliavin}
\begin{aligned}
\frac{\partial}{\partial \nu'}\Big|_{\nu'=\nu}\mathbb{E}_\xi^{\hat{\mathbb{P}}^{\nu'}}[H(X_T;\nu)]
&=\mathbb{E}_\xi^{\hat{\mathbb{P}}^{\nu}}\Big[H(X_T;\nu)\int_0^T \frac{\overline{\kappa}(X_s;\nu)}{\sigma(X_s)}\,d\hat{B}_s\Big]\\
&=\mathbb{E}_\xi^{\hat{\mathbb{P}}^{\nu}}\Big[\int_0^T D_s(H(X_T;\nu))\frac{\overline{\kappa}(X_s;\nu)}{\sigma(X_s)}\,ds\Big]\\
&=\mathbb{E}_\xi^{\hat{\mathbb{P}}^{\nu}}\Big[\int_0^T H_x(X_T;\nu)D_sX_T\frac{\overline{\kappa}(X_s;\nu)}{\sigma(X_s)}\,ds\Big]\\
&=\mathbb{E}_\xi^{\hat{\mathbb{P}}^{\nu}}\Big[H_x(X_T;\nu)Y_T\int_0^T \frac{\overline{\kappa}(X_s;\nu)}{Y_s}\,ds\Big],\\
\end{aligned} 
\end{equation}
where $Y$ is the first variation process of $X.$
For all the models in Section \ref{sec:u_max}, this probabilistic representation will be used to show that the partial derivative is bounded in $T$ on $[0,\infty).$\newline

\noindent {\bf Step IV}.
Since $\frac{f_\nu(T,\xi)}{f(T,\xi)}$ is bounded in $T$ on $[0,\infty)$ in Eq.\eqref{eqn:main_argu_p_T},
we finally obtain
$$\left|\frac{1}{T}\frac{\partial}{\partial \nu}\ln p_T+\frac{\partial\lambda}{\partial \nu}
\right|\leq \frac{c}{T}$$
for some positive constant $c.$ 
In particular, 
$$\lim_{T\to\infty}\frac{1}{T}\frac{\partial}{\partial \nu}\ln p_T=-\frac{\partial\lambda}{\partial \nu}\,.$$
This implies that the influence of the risk tolerance  parameter $\nu$ on long-term investments is determined by the eigenvalue of the generator of the underlying Markov diffusion.

\begin{remark}
	The idea of deriving Eq.\eqref{eqn:main_Malliavin} in {\bf Step III} is from \cite[Section 6.1]{borovivcka2014shock}. They presented the integration-by-parts formula in Malliavin calculus to compute the shock elasticity, and Eq.\eqref{eqn:main_Malliavin} comes from the same method.   
\end{remark}

\section{Utility-maximizing portfolios}
\label{sec:u_max}

We cover several models, 
namely the Black--Scholes model, the Ornstein--Uhlenbeck (OU) process, the Cox--Ingersoll--Ross (CIR) model, the $3/2$ model,
and a quadratic drift model.
Throughout this section, we assume the short rate is a constant $r$ (and thus, $p_T$ in Eq.\eqref{eqn:p_T} will be used).

\subsection{Black--Scholes model}

As a motivating example, consider a constant proportion portfolio
when the underlying market follows the Black--Scholes model.
Assume that the short rate is a constant $r\ge0$ and the stock price follows
$$dS_t=\mu S_t\,dt+\sigma S_t\,dZ_t$$
for $\mu\in\mathbb{R},$ $\sigma>0.$ 
With the initial capital $\omega>0,$ 
it is known that 
the optimal expected utility is
\begin{equation} 
\begin{aligned}
-u_T:=\max_{\pi}\mathbb{E}^{\bf P}[U(\Pi_T^\pi)]
=\frac{\omega^\nu}{\nu} e^{(r+\frac{(\mu-r)^2}{2(1-\nu)\sigma^2})\nu T}.
\end{aligned}
\end{equation}

We aim to investigate the influence of the risk aversion on the long-term investment.
We calculate the partial derivative
$$\frac{\partial}{\partial \nu}\ln u_T=\ln \omega-\frac{1}{\nu}+(r+\frac{(\mu-r)^2}{2(1-\nu)^2\sigma^2})T\,.$$
The partial derivative $\frac{\partial}{\partial \nu}\ln u_T$  grows linearly as $T\to\infty,$ and the linear growth rate is $r+\frac{(\mu-r)^2}{2(1-\nu)^2\sigma^2}.$

\subsection{OU process}

Under the physical measure ${\bf P},$ let the state process $X$ follow the OU process
\begin{equation}
\label{eqn:max_OU_dX}
dX_t= (b-kX_t)\,dt+ \sigma \,dZ_t\,,\;X_0=\xi
\end{equation} 
for $b,\xi\in\mathbb{R}$,  $k,\sigma>0$ and assume  that  $\theta'\theta(\cdot)=\cdot\,^2$ and $\rho'\theta(\cdot)=\overline{\rho}\,\cdot$ for some constant $\overline{\rho}\in\mathbb{R}$ (this holds, for example, $d=1$  and the state process is the market price of risk).   
Then, the $\mathbb{P}$-dynamics of $X$  given in Eq.\eqref{eqn:X_P} is 
\begin{equation} 
dX_t=(b-aX_t)\,dt+\sigma \,dB_t\,,\;X_0=\xi
\end{equation}
and
$p_T=\mathbb{E}^{{\mathbb{P}}}[e^{-q\int_0^T X_u^2\,du}]$,
where
$a=k-\frac{\nu\sigma\overline{\rho}}{1-\nu}$ and  
$q=-\frac{\nu(1-\nu+\nu\rho'\rho)}{2(1-\nu)^2}.$

We now apply the Hansen--Scheinkman decomposition
stated in Section \ref{sec:HS}.
Eq.\eqref{eqn:OU_quaratic_HS} gives
\begin{equation} 
\begin{aligned}
p_T=f(T,\xi) e^{-\frac{1}{2}\eta \xi^2-\ell \xi}e^{-\lambda T}, 
\end{aligned}
\end{equation}
where
$\alpha=\sqrt{a^2+2q\sigma^2},$ $\eta=\frac{\alpha-a}{\sigma^2},$ $\ell=\frac{b\eta}{\alpha},$  $\lambda=-\frac{1}{2}\sigma^2\ell^2+b\ell+\frac{1}{2}(\alpha-a)$,
and 
\begin{equation} 
f(t,x)=\mathbb{E}_x^{\hat{\mathbb{P}}}[e^{\frac{1}{2}\eta X_T^2+\ell X_T}]
\,,\; 0\le t\le T\,,\; x\in\mathbb{R}   
\end{equation} 
is the remainder function.
Under the measure $\hat{\mathbb{P}},$ the process $X$ satisfies
$$dX_t=(\delta-\alpha  X_t)\,dt+\sigma \,d\hat{B}_t$$
for $\delta=\frac{b}{\alpha}.$
It follows that  
\begin{equation}  \label{eqb:OU_partial_nu}
\frac{\partial}{\partial \nu}\ln p_T=-\frac{1}{2}\xi^2\frac{\partial\eta}{\partial \nu}-\xi\frac{\partial\ell}{\partial \nu}-T\frac{\partial\lambda}{\partial \nu}
+\frac{f_\nu(T,\xi)}{f(T,\xi)} \,.
\end{equation}

It suffices to investigate the term $f_\nu(T,\xi)$ since the other terms on the right-hand side of Eq.\eqref{eqb:OU_partial_nu} are easy to estimate.
Since only the parameters $\eta,$ $\ell,$ $\alpha,$ and $\delta$ depend on $\nu$ in the remainder function $f(T,\xi)$, using the chain rule, we know
that $$f_\nu(T,\xi)=f_\eta(T,\xi)\frac{\partial\eta}{\partial\nu}+f_\ell(T,\xi)\frac{\partial\ell}{\partial\nu}+f_\alpha(T,\xi)\frac{\partial\alpha}{\partial\nu}+f_\delta(T,\xi)\frac{\partial \delta}{\partial\nu}\,.$$
By Proposition \ref{prop:OU_quad_partial_nu}, 
the function $f(T,\xi)$ converges to a positive constant as $T\to\infty$, and four partial derivatives are bounded in $T.$ 
Therefore,  
$$\Big|\frac{1}{T}\ln p_T+\lambda\Big|\leq\frac{c}{T} \;\textnormal{ and }\;\Big|\frac{1}{T}\frac{\partial}{\partial \nu}\ln p_T+\frac{\partial\lambda}{\partial\nu}\Big|\leq\frac{c}{T}$$
for some positive constant $c.$
By using Eq.\eqref{eqn:-u_T}, we finally conclude that
$$\left|\frac{1}{T}\frac{\partial}{\partial \nu}\ln u_T-r-\frac{\rho'\rho\lambda}{(1-\nu+\rho'\rho\nu)^2}+\frac{1-\nu}{1-\nu+\nu \rho'\rho}\frac{\partial\lambda}{\partial\nu}\right|\leq\frac{c'}{T}$$
for some positive constant $c'$ and
\begin{equation}
\begin{aligned}
\frac{\partial\lambda}{\partial\nu}
&=(b-\sigma^2\ell)\frac{\partial\ell}{\partial\nu} +\frac{1}{2}\frac{\partial\alpha}{\partial\nu}-\frac{1}{2}\frac{\partial a}{\partial\nu}\\
&=-\Big(b(b-\sigma^2\ell)\Big(\frac{a^2}{\alpha^3\sigma}-\frac{1}{\alpha\sigma}\Big)+\frac{\sigma}{2}\Big(\frac{a}{\alpha}-1\Big)\Big)\frac{\overline{\rho}}{(1-\nu)^2}-\Big(\frac{ab(b-\sigma^2\ell)}{\alpha^2\sigma^2}+\frac{1}{2}\Big) \Big(\frac{1-\nu+2\rho'\rho\nu}{2(1-\nu)^3}\Big)
\end{aligned}
\end{equation}
which is obtained by direct calculation.

\begin{remark}
	The optimal expected utility has explicit solutions when the market price of risk is an affine model such as the OU process or the CIR model. 
	However, $\lim_{T\to\infty}\frac{\partial}{\partial \nu}\ln p_T$ is extremely complicated and  challenging to calculate from the explicit solutions.
	In this study, we adopt the Hansen--Scheinkman decomposition and Malliavin calculus so that it is much simpler to calculate the long-term sensitivity using our approach rather than using the explicit solutions. 
\end{remark}


\subsection{CIR model}
\label{sec:CIR}

Under the physical measure ${\bf P},$ let the state process $X$ follow the CIR model
\begin{equation}
\label{eqn:max_CIR_dX}
dX_t= (b-kX_t)\,dt+ \sigma\sqrt{X_t}\,dZ_t\,,\;X_0=\xi 
\end{equation} 
for  $k,\sigma,\xi>0,$ $b>\sigma^2/2$
and assume  that  $\theta'\theta(\cdot)=\cdot$ and $\rho'\theta(\cdot)=\overline{\rho}\sqrt{\,\cdot\,}$ for some constant $\overline{\rho}\in\mathbb{R}$ (this holds, for example, $d=1$  and the market price of risk is the square root of  the state process). 
Then, the $\mathbb{P}$-dynamics of $X$  given in Eq.\eqref{eqn:X_P} is 
\begin{equation} 
dX_t= (b-aX_t)\,dt+ \sigma\sqrt{X_t}\,dB_t\,,\;X_0=\xi 
\end{equation}
and
$p_T=\mathbb{E}^{{\mathbb{P}}}[e^{-q\int_0^T X_u\,du}]$,
where
$a=k-\frac{\nu\sigma\overline{\rho}}{1-\nu}$ and 
$q=-\frac{\nu(1-\nu+\nu\rho'\rho)}{2(1-\nu)^2}.$

We now apply the Hansen--Scheinkman decomposition.
Eq.\eqref{eqn:CIR_decompo} gives
\begin{equation} 
\begin{aligned}
p_T=f(T,\xi)e^{-\eta\xi}e^{-\lambda T}, 
\end{aligned}
\end{equation}
where
$\alpha:=\sqrt{a^2+2q\sigma^2},$ $\eta:=\frac{\alpha-a}{\sigma^2},$ $\lambda:=b\eta$,
and $f(t,x):=\mathbb{E}_x^{\hat{\mathbb{P}}}[e^{\eta X_T}]$ 
is the remainder function.
The $\hat{\mathbb{P}}$-dynamics of $X$ is
$$dX_t= (b-\alpha X_t)\,dt+\sigma\sqrt{X_t}\,d\hat{B}_t\,.$$
It follows that  
\begin{equation}  
\frac{\partial}{\partial \nu}\ln p_T=-\xi\frac{\partial\eta}{\partial \nu}-T\frac{\partial\lambda}{\partial \nu}
+\frac{f_\nu(T,\xi)}{f(T,\xi)} \,.
\end{equation}

It suffices to investigate the term $f_\nu(T,\xi)$ since the other terms on the right-hand side of the above-mentioned equality are easy to estimate.
Since only the parameters $\eta$ and $\alpha$ depend on $\nu$ in the remainder function, using the chain rule, we know that
$$f_\nu(T,\xi)=f_\eta(T,\xi)\frac{\partial\eta}{\partial\nu}+f_\alpha(T,\xi)\frac{\partial\alpha}{\partial\nu}\,.$$
By Proposition \ref{prop:app_OU_quad_partial_nu}, 
the function $f(T,\xi)$ converges to a positive constant as $T\to\infty$, and two partial derivatives $f_\eta(T,\xi)$ and $f_\alpha(T,\xi)$ are bounded in $T.$ 
Therefore, 
$$\Big|\frac{1}{T}\ln p_T+\lambda\Big|\leq\frac{c}{T} \;\textnormal{ and }\;\Big|\frac{1}{T}\frac{\partial}{\partial \nu}\ln p_T+\frac{\partial\lambda}{\partial\nu}\Big|\leq\frac{c}{T}$$
for some positive constant $c.$
By using Eq.\eqref{eqn:-u_T}, we finally conclude that
$$\left|\frac{1}{T}\frac{\partial}{\partial \nu}\ln u_T-r-\frac{\rho'\rho\lambda}{(1-\nu+\rho'\rho\nu)^2}+\frac{1-\nu}{1-\nu+\nu \rho'\rho}\frac{\partial\lambda}{\partial\nu}\right|\leq\frac{c'}{T}$$
for some positive constant $c'$ and
$$\frac{\partial\lambda}{\partial\nu}
=\frac{b}{\sigma}\Big(\Big(\frac{a}{\alpha}-1\Big)\frac{\partial a}{\partial\nu}+\frac{\sigma^2}{\alpha}\frac{\partial q}{\partial \nu}\Big)
=\Big(1-\frac{a}{\alpha}\Big)\frac{b\overline{\rho}}{(1-\nu)^2}-\frac{b\sigma(1-\nu+2\rho'\rho\nu)}{2\alpha(1-\nu)^3}$$
which is obtained by direct calculation.

\subsection{$3/2$ model}

Under the physical measure ${\bf P},$ let the state process $X$ follow   the $3/2$ model
\begin{equation}\label{eqn:3/2}
dX_t= (b-kX_t)X_t\,dt+ \sigma X_t^{3/2}\,dZ_t\,,\;X_0=\xi
\end{equation} 
for $b,k,\sigma,\xi>0$ and assume  that  $\theta'\theta(\cdot)=\cdot$ and $\rho'\theta(\cdot)=\overline{\rho}\sqrt{\,\cdot\,}$ for some constant $\overline{\rho}\in\mathbb{R}$ (this holds, for example, $d=1$  and the market price of risk is the square root of  the state process).  
Then, the $\mathbb{P}$-dynamics of $X$  given in Eq.\eqref{eqn:X_P} is 
\begin{equation} 
dX_t= (b-aX_t)\,dt+ \sigma X_t^{3/2}\,dB_t\,,\;X_0=\xi 
\end{equation}
and
$p_T=\mathbb{E}^{{\mathbb{P}}}[e^{-q\int_0^T X_u\,du}]$,
where
$a=k-\frac{\nu\sigma\overline{\rho}}{1-\nu}$ and 
$q=-\frac{\nu(1-\nu+\nu\rho'\rho)}{2(1-\nu)^2}.$

We now apply the Hansen--Scheinkman decomposition.
Eq.\eqref{eqn:3/2_HS_decompo} gives
\begin{equation} 
\begin{aligned}
p_T=f(T,\xi)\xi^{-\eta}e^{-\lambda T},
\end{aligned}
\end{equation}
where 
$\eta:=\frac{\sqrt{(a+\sigma^2/2)^2+2q\sigma^2}-(a+\sigma^2/2)}{\sigma^2},$   $\lambda:=b\eta$,
and $f(t,x)=\mathbb{E}_x^{\hat{\mathbb{P}}}[X_t^{\eta}]$ 
is the remainder function.
The $\hat{\mathbb{P}}$-dynamics of $X$ is
\begin{equation}
\begin{aligned}
dX_t=(b-\alpha X_t)X_t\,dt+\sigma {X_t}^{3/2}\,d\hat{B}_t
\end{aligned}
\end{equation}
for $\alpha:=a+\sigma^2\eta.$
It follows that  
\begin{equation}  
\frac{\partial}{\partial \nu}\ln p_T= -\ln\xi\frac{\partial\eta}{\partial \nu}-T\frac{\partial\lambda}{\partial \nu}
+\frac{f_\nu(T,\xi)}{f(T,\xi)} \,.
\end{equation}

It suffices to investigate the term $f_\nu(T,\xi)$ since the other terms on the right-hand side of the above-mentioned equality are easy to estimate.
Since only the parameters $\eta$ and $\alpha$ depend on $\nu$ in the remainder function, using the chain rule, we know that
$$f_\nu(T,\xi)=f_\eta(T,\xi)\frac{\partial\eta}{\partial\nu}+f_\alpha(T,\xi)\frac{\partial\alpha}{\partial\nu}\,.$$
The function $f(T,\xi)$ converges to a positive constant as $T\to\infty,$ and two partial derivatives $f_\eta(T,\xi)$ and $f_\alpha(T,\xi)$ are bounded in $T$
by Proposition \ref{prop:app_3/2}.
Therefore,
$$\Big|\frac{1}{T}\ln p_T+\lambda\Big|\leq\frac{c}{T} \;\textnormal{ and }\;\Big|\frac{1}{T}\frac{\partial}{\partial \nu}\ln p_T+\frac{\partial\lambda}{\partial\nu}\Big|\leq\frac{c}{T}$$
for some positive constant $c.$
By using Eq.\eqref{eqn:-u_T}, we finally conclude that
$$\left|\frac{1}{T}\frac{\partial}{\partial \nu}\ln u_T-r-\frac{\rho'\rho\lambda}{(1-\nu+\rho'\rho\nu)^2}+\frac{1-\nu}{1-\nu+\nu \rho'\rho}\frac{\partial\lambda}{\partial\nu}\right|\leq\frac{c'}{T}$$
for some positive constant $c'$ and
\begin{equation}
\begin{aligned}
\frac{\partial\lambda}{\partial\nu}
&=\frac{b}{\sqrt{(a+\sigma^2/2)^2+2q\sigma^2}}\Big(-\eta\frac{\partial a}{\partial\nu}+\frac{\partial q}{\partial \nu}\Big)\\
&=\frac{b}{(1-\nu)^2\sqrt{(a+\sigma^2/2)^2+2q\sigma^2}}\Big(\eta\sigma\overline{\rho}-\frac{1-\nu+2\rho'\rho\nu}{2(1-\nu)}\Big)\,. 
\end{aligned}
\end{equation}

\subsection{Quadratic drift model}

Under the physical measure ${\bf P},$ let the state process $X$ follow a quadratic drift model
\begin{equation}\label{eqn:garch}
dX_t= (b-kX_t^2)\,dt+ \sigma X_t\,dZ_t\,,\;X_0=\xi\,
\end{equation} 
for $b,k,\xi>0,$  $\sigma\neq0$
and assume  that  $\theta'\theta(\cdot)=\cdot\,^2$ and $\rho'\theta(\cdot)=\overline{\rho}\,\cdot$ for some constant $\overline{\rho}\in\mathbb{R}$ (this holds, for example, $d=1$  and the state process is the market price of risk). 
There is a unique strong solution to the above SDE  \cite[Proposition 2.1]{carr2019lognormal}.
Then, under the measure $\mathbb{P},$  the process $X$ satisfies
\begin{equation} 
dX_t= (b-aX_t^2)\,dt+ \sigma X_t\,dB_t\,,\;X_0=\xi 
\end{equation}
and
$p_T=\mathbb{E}^{{\mathbb{P}}}[e^{-q\int_0^T X_u^2\,du}]$,
where
$a=k-\frac{\nu\sigma\overline{\rho}}{1-\nu}$ and 
$q=-\frac{\nu(1-\nu+\nu\rho'\rho)}{2(1-\nu)^2}.$

We now apply the Hansen--Scheinkman decomposition.
Eq.\eqref{eqn:quad_HS_decompo} gives
\begin{equation} 
\begin{aligned}
p_T=f(T,\xi)e^{-\eta\xi}e^{-\lambda  T},
\end{aligned}
\end{equation}
where 
$\alpha:=\sqrt{a^2+2q\sigma^2},$ $\eta:=\frac{\alpha-a}{\sigma^2},$ $\lambda:=b\eta$,
and $f(T,\xi):=\mathbb{E}_\xi^{\hat{\mathbb{P}}}[e^{\eta X_T}]$ 
is the remainder function.
The $\hat{\mathbb{P}}$-dynamics of $X$ is
\begin{equation}
\begin{aligned}
dX_t=(b-\alpha X_t^2)\,dt+\sigma X_t\,d\hat{B}_t\,.
\end{aligned}
\end{equation} 
It follows that  
\begin{equation}  
\frac{\partial}{\partial \nu}\ln p_T= -\ln\xi\frac{\partial\eta}{\partial \nu}-T\frac{\partial\lambda}{\partial \nu}
+\frac{f_\nu(T,\xi)}{f(T,\xi)} \,.
\end{equation}

It suffices to investigate the term $f_\nu(T,\xi)$ since the other terms on the right-hand side of the above-mentioned equality are easy to estimate.
Since only the parameters $\eta$ and $\alpha$ depend on $\nu$ in the remainder function, using the chain rule, we know that
$$f_\nu(T,\xi)=f_\eta(T,\xi)\frac{\partial\eta}{\partial\nu}+f_\alpha(T,\xi)\frac{\partial\alpha}{\partial\nu}\,.$$
The function $f(T,\xi)$ converges to a positive constant as $T\to\infty,$ and two partial derivatives $f_\eta(T,\xi)$ and $f_\alpha(T,\xi)$ are bounded in $T$
by Proposition \ref{prop:quad_model_remain}.
Therefore,
$$\Big|\frac{1}{T}\ln p_T+\lambda\Big|\leq\frac{c}{T} \;\textnormal{ and }\;\Big|\frac{1}{T}\frac{\partial}{\partial \nu}\ln p_T+\frac{\partial\lambda}{\partial\nu}\Big|\leq\frac{c}{T}$$
for some positive constant $c.$
By using Eq.\eqref{eqn:-u_T}, we finally conclude that
$$\left|\frac{1}{T}\frac{\partial}{\partial \nu}\ln u_T-r-\frac{\rho'\rho\lambda}{(1-\nu+\rho'\rho\nu)^2}+\frac{1-\nu}{1-\nu+\nu \rho'\rho}\frac{\partial\lambda}{\partial\nu}\right|\leq\frac{c'}{T}$$
for some positive constant $c'$ and
\begin{equation}
\begin{aligned}
\frac{\partial\lambda}{\partial\nu}
&=\frac{b}{\alpha}\Big(-\eta\frac{\partial a}{\partial\nu}+\frac{\partial q}{\partial \nu}\Big)\\
&=\frac{b}{\alpha(1-\nu)^2}\Big(\eta\sigma\overline{\rho}-\frac{1-\nu+2\rho'\rho\nu}{2(1-\nu)}\Big)\,. 
\end{aligned}
\end{equation}

Figure \ref{fig} displays comparative analysis between the four models.
Two graphs show the partial derivative $\frac{\partial\lambda}{\partial\nu}$  as a function of $\nu$ and $\mu,$ respectively.
The model parameters are given as 
$b=0.16,$ $\sigma=0.8,$ $k=2$ (for the first graph), $\nu=-2$ (for the second graph) and  $\rho=-0.5,$
\begin{figure}\begin{center} \includegraphics[trim={0cm 0cm 0cm 0cm},clip,scale = 0.3]{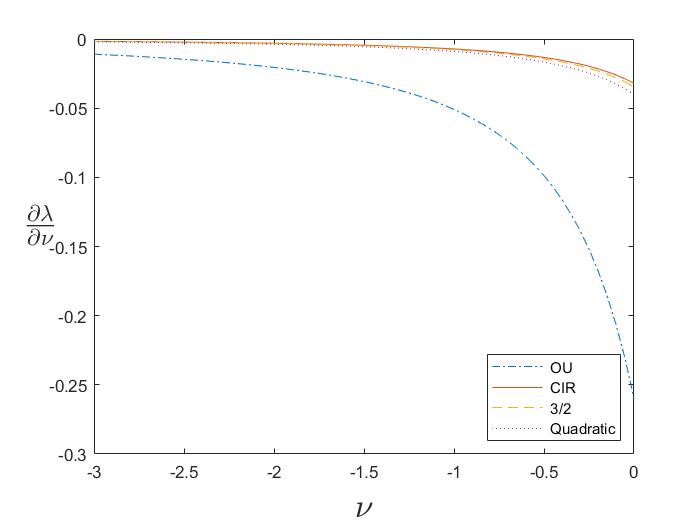}
\includegraphics[trim={0cm 0cm 0cm 0cm},clip,scale = 0.3]{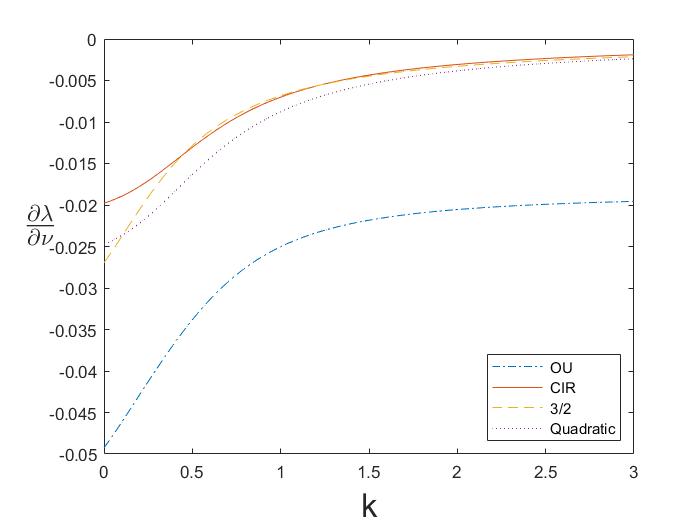} 
\caption{Comparative analysis between the different models.}\label{fig}\end{center} \end{figure}

\section{Conclusion}
\label{sec:con}

This study investigated the influence of risk tolerance on the expected utility in the long run.
We focused on the power utility function of the form
$$U(x)=-x^{\nu}$$
for $\nu<0$, where the parameter $\nu$ represents how an investor measures the degree of his/her risk tolerance.
We considered utility-maximizing portfolios and demonstrated the influence of small changes in the parameter $\nu$ on the expected utility of the portfolios in the long run.

The main methodology for this analysis involved a combination of the Hansen--Scheinkman decomposition and the Malliavin calculus technique. First, we transformed the expected utility into the expectation form 
$$p_T=\mathbb{E}_\xi[e^{-\int_0^Tq(X_s)\,ds}h(X_T)]$$
for some Markov diffusion process $X=(X_t)_{0\le t\le T}$ with $X_0=\xi$ and some measurable functions $q$ and $h.$
Using the Hansen--Scheinkman decomposition, 
the expectation $p_T$ was written as
$$p_T=\phi(\xi)e^{-\lambda T}
f(T,\xi)$$ 
for a real number $\lambda,$ a positive function $\phi$, and a measurable function $f,$ which depend on the parameter $\nu.$

The influence of risk tolerance on the long-term expected utility was obtained from the above-mentioned Hansen--Scheinkman decomposition.
Under the condition that $\frac{1}{T}\frac{\partial}{\partial\nu}\ln f(T,\xi)$
is bounded in $T$ on $[0,\infty),$ we showed that
$$\left|\frac{\partial}{\partial\nu}\ln p_T+\frac{\partial\lambda}{\partial\nu}\right|\leq \frac{c}{T}$$
for some positive constant $c.$
The influence of risk tolerance
is asymptotically equal to the partial derivative of $-\lambda$ with respect to $\nu,$ which is the main conclusion of this study.
To verify that $\frac{1}{T}\frac{\partial}{\partial\nu}\ln f(T,\xi)$
is bounded in $T$ on $[0,\infty),$ the Malliavin calculus method was used under 
several market models, namely the Ornstein--Uhlenbeck process, the CIR process, the $3/2$ model, and a quadratic drift model. \newline

\noindent\textbf{Acknowledgments.}\\ 
Hyungbin Park was supported by the National Research Foundation of Korea (NRF) grants funded by the Ministry of Science and ICT (No. 2017R1A5A1015626, No. 2018R1C1B5085491 and No. 2021R1C1C1011675) 
and the Ministry of Education   (No. 2019R1A6A1A10073437) through the Basic Science Research Program.

\appendices

\section{OU process with quadratic killing rate}

Assume that a process $X$ satisfies
\begin{equation} \label{eqn:OU_under_P_II}
dX_t=(b-aX_t)\,dt+\sigma \,dB_t\,,\;X_0=\xi,
\end{equation}
where $b\in\mathbb{R},$ $a,\sigma>0.$  
For $q>-\frac{a^2}{2\sigma^2},$  consider the  expectation 
$$p_T:=\mathbb{E}_\xi^{{\mathbb{P}}}[e^{-q\int_0^T X_u^2\,du}]\,.$$

The corresponding operator is
$$\mathcal{P}_Th(x)=\mathbb{E}_x^{{\mathbb{P}}}[e^{-q\int_0^T X_u^2\,du}h(X_T)]\,,$$
and it can be shown that 
$$(\lambda,\phi(x)):=\Big(-\frac{1}{2}\sigma^2\ell^2+b\ell+\frac{1}{2}(\alpha-a),e^{-\frac{1}{2}\eta x^2-\ell x}\Big)$$
is an eigenpair, 
where
$$\alpha:=\sqrt{a^2+2q\sigma^2}\,,\;\eta:=\frac{\alpha-a}{\sigma^2}\,,\;\ell:=\frac{b\eta}{\alpha}\,.$$ 
The eigen-measure $\hat{\mathbb{P}}$ is defined on   $\mathcal{F}_T$ as
\begin{equation}
\frac{d\hat{\mathbb{P}}}{d\mathbb{P}}
=e^{-\frac{1}{2}\sigma^2\int_0^T(\eta X_s+\ell)^2\,ds -\sigma\int_0^T (\eta X_s+\ell) \,dB_s}\,.
\end{equation} 
The expectation $p_T$ can be expressed as 
\begin{equation} \label{eqn:OU_quaratic_HS}
\begin{aligned}
p_T&=\mathbb{E}_\xi^\mathbb{P}[e^{-q\int_0^TX_s^2\,ds}]\\
&=\mathbb{E}_\xi^{\hat{\mathbb{P}}}[e^{\frac{1}{2}\eta X_T^2+\ell X_T}] e^{-\frac{1}{2}\eta \xi^2-\ell \xi}e^{-\lambda T}=f(T,\xi) e^{-\frac{1}{2}\eta \xi^2-\ell \xi}e^{-\lambda T}, 
\end{aligned}
\end{equation}
where
\begin{equation} 
f(t,x)=\mathbb{E}_x^{\hat{\mathbb{P}}}[e^{\frac{1}{2}\eta X_T^2+\ell X_T}]
\,,\; 0\le t\le T\,,\; x\in\mathbb{R}   
\end{equation}
is the remainder function.
The process
$$\hat{B}_t=B_t+\sigma\int_0^t (\eta X_s+\ell) \,ds\,,\;0\le t\le T$$
is a $\hat{\mathbb{P}}$-Brownian motion
and    $X$ follows
$$dX_t=(\delta-\alpha  X_t)\,dt+\sigma \,d\hat{B}_t\,$$
for $\delta:=\frac{b}{\alpha}.$

We study the large-time asymptotic behavior of the sensitivity of the remainder function $f.$

\begin{prop} \label{prop:OU_quad_partial_nu}
	Suppose that $X$ follows 
	\begin{equation}
	\label{eqn:app_OU_quad}
	dX_t=(\delta-\alpha  X_t)\,dt+\sigma \,d\hat{B}_t\,,\;X_0=\xi
	\end{equation} 
	for  $\delta,\xi\in\mathbb{R}$ and $\alpha,\sigma>0.$ Define 
	$$f(T,\xi)=\mathbb{E}_\xi^{\hat{\mathbb{P}}}[e^{\frac{1}{2}\eta X_T^2+\ell X_T}]
	\,,\; T\ge0 \,,$$
	for $\eta<\frac{2\alpha}{\sigma^2}$ and $\ell\in\mathbb{R}.$
	Then,
	\begin{equation}\label{eqn:quand_OU_conv}
	f(T,\xi)\to \int e^{\frac{1}{2}\eta x^2+\ell x}\,d\pi(x)
	\end{equation}
	as $T\to\infty$, where $\pi$ is the invariant distribution of $X.$ The partial derivatives $f_\eta(T,\xi),$ $f_\ell(T,\xi),$
	$f_\alpha(T,\xi),$ and $f_\delta(T,\xi)$ are bounded in $T$ on $[0,\infty).$
\end{prop}

\begin{proof}  
	Observe that
	the density function of $X_T$ with $X_0=\xi$ is 
	$$z(x;T):=\frac{1}{\Sigma_T\sqrt{2\pi }}e^{-\frac{1}{2}\frac{(x-m_T)^2}{\Sigma_T^2}}\,,$$
	where $m_T:=\xi e^{-\alpha T}+\frac{\delta}{\alpha}(1-e^{-\alpha T})$ is the mean and $\Sigma_T^2:=\frac{\sigma^2}{2\alpha}(1-e^{-2\alpha T})$ is the variance.
	Then, it is clear that
	$$f(T,\xi)=\mathbb{E}_\xi^{\hat{\mathbb{P}}}[e^{\frac{1}{2}\eta X_T^2+\ell X_T}]=\int_{\mathbb{R}}e^{\frac{1}{2}\eta x^2+\ell  x}z(x;T)\,dx\to \int_{\mathbb{R}}e^{\frac{1}{2}\eta x^2+\ell x}\,d\pi(x)$$
	as $T\to\infty$, 
	where
	$$d\pi(x):=\frac{1}{\sqrt{\pi  \sigma^2/\alpha}}e^{-\frac{(x-\delta/\alpha)^2}{\sigma^2/\alpha}}\,dx\,.$$
	This proves Eq.\eqref{eqn:quand_OU_conv}.
	
	We now show that $f_\eta(T,\xi)$ is bounded in $T$ on $[0,\infty).$ This is direct from 
	\begin{equation}\label{eqn:f_eta_OU_2}
	\begin{aligned}
	f_\eta(T,\xi)
	&=\frac{\partial}{\partial\eta}\mathbb{E}_\xi^{\hat{\mathbb{P}}}[e^{\frac{1}{2}\eta X_T^2+\ell X_T}]\\
	&=\mathbb{E}_\xi^{\hat{\mathbb{P}}}\Big[\frac{\partial}{\partial\eta}\Big(e^{\frac{1}{2}\eta X_T^2+\ell X_T}\Big)\Big]\\
	&=\frac{1}{2}\mathbb{E}_\xi^{\hat{\mathbb{P}}}[X_T^2e^{\frac{1}{2}\eta X_T^2+\ell X_T}]\to \frac{1}{2}\int_{\mathbb{R}} y^2e^{\frac{1}{2}\eta y^2+\ell  y}z(y;\infty)\,dy 
	\end{aligned}
	\end{equation}
	since the limit is a finite number.
	Using the same method, we can show that  
	$f_\ell(T,\xi)$ is bounded in $T$ on $[0,\infty).$

	We show that $f_\alpha(T,\xi)$ is bounded in $T$ on $[0,\infty).$
	Define $H(x)=e^{\frac{1}{2}\eta x^2+\ell x}$ for notational simplicity. 
	By Propositions \ref{prop:Malliavin_chain}, \ref{prop:IBP}, and \ref{prop:rho},  we have
	\begin{equation}\label{eqn:quad_OU_malliavin}
	\begin{aligned} 
	f_\alpha(T,\xi)=\frac{\partial}{\partial\alpha}\mathbb{E}_\xi^{\hat{\mathbb{P}}}[H(X_T)]
	&=-\frac{1}{\sigma}\mathbb{E}_\xi^{\hat{\mathbb{P}}}\Big[H(X_T)\int_0^TX_s\,d\hat{B}_s\Big]\\
	&=-\frac{1}{\sigma}\mathbb{E}_\xi^{\hat{\mathbb{P}}}\Big[\int_0^TD_s(H(X_T))X_s\,ds\Big]\\
	&=-\frac{1}{\sigma}\mathbb{E}_\xi^{\hat{\mathbb{P}}}\Big[\int_0^TH'(X_T)(D_sX_T)X_s\,ds\Big]\,.
	\end{aligned}
	\end{equation}
	Considering that the Malliavin derivative of $X_T$ is
	$D_sX_T=\sigma e^{-a(T-s)}$ for $s\le T,$
	we have  
	\begin{equation} 
	\begin{aligned} 
	|f_\delta(T,\xi)|
	&\leq e^{-aT}\mathbb{E}_\xi^{\hat{\mathbb{P}}}\Big[\int_0^T |H'(X_T)|X_se^{as}\,ds\Big]=e^{-aT}\int_0^T \mathbb{E}_\xi^{\hat{\mathbb{P}}}[|H'(X_T)|X_s]e^{as}\,ds\,.
	\end{aligned}
	\end{equation}
	Choose $u>1$ such that $\frac{1}{2}\eta u<\frac{\alpha}{\sigma^2}$ and $v>1$ such that $1/u+1/v=1.$ Then, 
	$$|f_\delta(T,\xi)|\leq e^{-aT}\int_0^T \mathbb{E}_\xi^{\hat{\mathbb{P}}}[|H'(X_T)|X_s]e^{as}\,ds
	\leq e^{-aT}\int_0^T \mathbb{E}_\xi^{\hat{\mathbb{P}}}[|H'(X_T)|^u]^{1/u}\mathbb{E}_\xi^{\hat{\mathbb{P}}}[X_s^v]^{1/v}e^{as}\,ds\,.$$
	Using $H'(x)=e^{\frac{1}{2}\eta x^2+\ell x}(\eta x+\ell)$ and the density function of $X_T,$
	it is easy to check that   $\mathbb{E}_\xi^{\hat{\mathbb{P}}}[|H'(X_T)|^u]$ 
	and $\mathbb{E}_\xi^{\hat{\mathbb{P}}}[X_s^v]$
	are bounded in $T$ and $s,$ respectively.
	This gives the desired result.
	Using the same method, we can show that  
	$f_\delta(T,\xi)$ is bounded in $T$ on $[0,\infty).$
\end{proof}

\section{CIR model}
\label{app:CIR}

Let $X$ be the  CIR model given as 
\begin{equation} \label{eqn:CIR_under_P}
dX_t=(b-aX_t)\,dt+\sigma \sqrt{X_t}\,dB_t\,,\;X_0=\xi,
\end{equation}
where $a,\sigma,\xi>0$ and $2b\ge \sigma^2.$
For 
\begin{equation}
\label{eqn:CIR_q}
q> -\frac{a^2}{2\sigma^2}\,,
\end{equation} 
consider the  expectation 
$$p_T:=\mathbb{E}_\xi^{{\mathbb{P}}}[e^{-q\int_0^T X_u\,du}]\,.$$ 
The corresponding  operator is
$$\mathcal{P}_Th(x)=\mathbb{E}_x^{{\mathbb{P}}}[e^{-q\int_0^T X_u\,du}h(X_T)]\,.$$
Following \cite{qin2016positive}, we know that 
$$(\lambda,\phi(x)):=(b\eta, e^{-\eta x})$$
is an eigenpair of this operator,
where  
$$\alpha:=\sqrt{a^2+2q\sigma^2}\,,\;
\eta:=\frac{\alpha-a}{\sigma^2}\,.$$ 
The eigen-measure $\hat{\mathbb{P}}$ is defined on   $\mathcal{F}_T$ as
\begin{equation}
\frac{d\hat{\mathbb{P}}}{d\mathbb{P}}
=e^{-\frac{1}{2}\sigma^2\eta^2\int_0^T X_s\,ds-\sigma\eta \int_0^T\sqrt{X_s}\,dB_s}.
\end{equation} 
It is easy to check that a local martingale $(e^{-\frac{1}{2}\sigma^2\eta^2\int_0^t X_s\,ds-\sigma\eta \int_0^t\sqrt{X_s}\,dB_s})_{0\le t\le T}$ is a martingale.
The expectation $p_T$ can be expressed as 
\begin{equation}\label{eqn:CIR_decompo}
\begin{aligned}
p_T&=\mathbb{E}_\xi^\mathbb{P}[e^{-q\int_0^TX_s\,ds}]
=\mathbb{E}_\xi^{\hat{\mathbb{P}}}[e^{\eta X_T}] e^{-\eta\xi}e^{-\lambda T}=f(T,\xi)e^{-\eta\xi}e^{-\lambda T}, 
\end{aligned}
\end{equation}
where \begin{equation}\label{eqn:CIR_reminder}
f(t,x)=\mathbb{E}_x^{\hat{\mathbb{P}}}[e^{\eta X_t}]
\,,\; 0\le t\le T\,,\; x>0   
\end{equation}
is the remainder function.
Note that 
the expectation  $f(t,x)<\infty$ by Lemma \ref{lem:CIR_MMG}
because  $\eta=\frac{\alpha-a}{\sigma^2}<\frac{2\alpha}{\sigma^2}.$
The process
$$\hat{B}_t=B_t+\sigma\eta\int_0^t\sqrt{X_s}\,ds\,,\;0\le t\le T$$
is a $\hat{\mathbb{P}}$-Brownian motion
and    $X$ follows
$$dX_t= (b-\alpha X_t)\,dt+\sigma\sqrt{X_t}\,d\hat{B}_t\,.$$

We study the large-time asymptotic behavior of the
remainder function $f$ and its sensitivity  with respect to the parameters $\alpha$ and $\eta.$

\begin{prop} \label{prop:app_OU_quad_partial_nu}
	Suppose that $X$ follows 
	\begin{equation}
	\label{eqn:app_CIR_quad}
	dX_t= (b-\alpha X_t)\,dt+\sigma\sqrt{X_t}\,d\hat{B}_t\,,\;X_0=\xi
	\end{equation} 
	for  $\alpha,\sigma,\xi>0$ and $2b\ge \sigma^2.$ Define 
	$$f(T,\xi)=\mathbb{E}_\xi^{\hat{\mathbb{P}}}[e^{\eta X_T}]
	\,,\; T\ge0 \,,$$
	for  $\eta<\frac{2\alpha}{\sigma^2}.$ 
	Then, 
	\begin{equation}\label{eqn:quand_CIR_conv}
	f(T,\xi)\to \int e^{\eta x}\,d\pi(x)
	\end{equation}
	as $T\to\infty$, where $\pi$ is the invariant distribution of $X.$ The partial derivatives $f_\eta(T,\xi)$ and
	$f_\alpha(T,\xi)$ are bounded in $T$ on $[0,\infty).$
\end{prop}

\begin{proof}
	It is easy to prove Eq.\eqref{eqn:quand_CIR_conv} by considering the density function of $X;$ thus, we omit the proof.  
	Consider the partial derivative $f_\eta(T,\xi).$
	Choose any $\gamma$ with $\eta=\frac{\alpha-a}{\sigma^2}<\gamma<\frac{\alpha}{\sigma^2},$ then   there is a positive constant $c_\gamma$ such that
	$e^{\eta x}x\leq c_\gamma e^{\gamma x}$ for $x>0.$
	Observe that
	$$f_\eta(T,\xi)=\frac{\partial}{\partial\eta}\mathbb{E}_\xi^{\hat{\mathbb{P}}}[e^{\eta X_T}]= \mathbb{E}_\xi^{\hat{\mathbb{P}}}\Big[\frac{\partial}{\partial\eta}e^{\eta X_T}\Big]
	= \mathbb{E}_\xi^{\hat{\mathbb{P}}}[e^{\eta X_T}X_T]\,;$$
	thus,
	$$|f_\eta(T,\xi)|\leq \mathbb{E}_\xi^{\hat{\mathbb{P}}}[e^{\eta X_T}X_T]\leq c_\gamma\mathbb{E}_\xi^{\hat{\mathbb{P}}}[e^{\gamma X_T}] \,.$$
	By Lemma \ref{lem:CIR_MMG}, the expectation $\mathbb{E}_\xi^{\hat{\mathbb{P}}}[e^{\gamma X_T}]$
	is bounded in $T$ on $[0,\infty)$, and this gives the desired result.

	Now, we show that the partial derivative $f_\alpha(T,\xi)$
	is bounded in $T$ on $[0,\infty).$ 
	Define $H(x)=e^{\eta x}$ for notational simplicity.
	By Propositions \ref{prop:Malliavin_chain}, \ref{prop:IBP}, and \ref{prop:rho},  we have
	\begin{equation} 
	\begin{aligned}
	f_\alpha(T,\xi)=\frac{\partial}{\partial\alpha}\mathbb{E}_\xi^{\hat{\mathbb{P}}}[H(X_T)]
	&=-\frac{1}{\sigma}\mathbb{E}_\xi^{\hat{\mathbb{P}}}\Big[H(X_T)\int_0^T\sqrt{X_s}\,d\hat{B}_s\Big]\\
	&=-\frac{1}{\sigma}\mathbb{E}_\xi^{\hat{\mathbb{P}}}\Big[ \int_0^TD_s(H(X_T))\sqrt{X_s}\,ds\Big] \\
	&=-\frac{1}{\sigma}\mathbb{E}_\xi^{\hat{\mathbb{P}}}\Big[ \int_0^TH'(X_T) (D_sX_T)\sqrt{X_s}\,ds\Big]\\
	&=-\mathbb{E}_\xi^{\hat{\mathbb{P}}}\Big[ \int_0^TH'(X_T) e^{\int_s^T\big(-\frac{\alpha}{2}-(\frac{b}{2}-\frac{\sigma^2}{8})\frac{1}{X_u}\big)\,du}\sqrt{X_T}\sqrt{X_s}\,ds\Big]\,.
	\end{aligned}
	\end{equation} 
	The last equality is from 
	$$D_sX_T=\sigma e^{\int_s^T\big(-\frac{\alpha}{2}-(\frac{b}{2}-\frac{\sigma^2}{8})\frac{1}{X_u}\big)\,du}\sqrt{X_T}\,,$$
	which is obtained by $D_sX_T=\sigma\sqrt{X_s}\frac{Y_T}{Y_s}$
	for the first variation process $Y$ of $X.$ This can be obtained from the work of \cite{alos2008malliavin} (note that Proposition \ref{prop:Malliavin_first_variation} cannot be applied here because the coefficients in Eq.\eqref{eqn:app_CIR_quad} do not have bounded derivatives).
	Then,
	\begin{equation} 
	\begin{aligned}
	|f_\alpha(T,\xi)|
	&\le\mathbb{E}_\xi^{\hat{\mathbb{P}}}\Big[ \int_0^T|H'(X_T)| e^{\int_s^T\big(-\frac{\alpha}{2}-(\frac{b}{2}-\frac{\sigma^2}{8})\frac{1}{X_u}\big)\,du}\sqrt{X_T}\sqrt{X_s}\,ds\Big]\\
	&\le\mathbb{E}_\xi^{\hat{\mathbb{P}}}\Big[ \int_0^T|H'(X_T)| e^{-\frac{\alpha}{2}(T-s)}\sqrt{X_T}\sqrt{X_s}\,ds\Big] \quad\Big(\because \frac{b}{2}-\frac{\sigma^2}{8}>0\Big)\\
	&\le  \int_0^T\mathbb{E}_\xi^{\hat{\mathbb{P}}}[|H'(X_T)|\sqrt{X_T}\sqrt{X_s}]e^{-\frac{\alpha}{2}(T-s)}\,ds \\
	&\le  \int_0^T\mathbb{E}_\xi^{\hat{\mathbb{P}}}[(H'(X_T))^2X_T]^{1/2} \mathbb{E}_\xi^{\hat{\mathbb{P}}}[X_s]^{1/2}e^{-\frac{\alpha}{2}(T-s)}\,ds.  
	\end{aligned}
	\end{equation}
	Since $(H'(X_T))^2X_T= \eta^2 X_Te^{2\eta X_T}$ and $2\eta<\frac{2\alpha}{\sigma^2},$ by Lemma \ref{lem:CIR_MMG}, the expectation $\mathbb{E}_\xi^{\hat{\mathbb{P}}}[(H'(X_T))^2X_T]$ is bounded in $T$ on $[0,\infty).$
	It is clear that $\mathbb{E}_\xi^{\hat{\mathbb{P}}}[X_s]$
	is bounded in $s$ on $[0,\infty).$ 
	Thus,
	\begin{equation} 
	\begin{aligned}
	|f_\alpha(T,\xi)|
	\le  c \int_0^T e^{-\frac{\alpha}{2}(T-s)}\,ds  \leq \frac{2c}{\alpha} 
	\end{aligned}
	\end{equation}
	for some positive constant $c,$ 
	which gives the desired result.
\end{proof}

\begin{lemma}\label{lem:CIR_MMG} 
	Let $X$ be a solution of
	$$dX_t= (b-\alpha X_t)\,dt+\sigma\sqrt{X_t}\,d{B}_t\,,\;X_0=x\,,$$
	where $\alpha,\sigma,x>0$ and $2b>\sigma^2.$ 
	For $\gamma<2\alpha/\sigma^2$, we have
	$$\mathbb{E}_x^{\mathbb{P}}[e^{\gamma X_T}]=\Big(\frac{1}{1-\gamma c(T)}\Big)^{2b/\sigma^2}e^{\frac{\gamma }{1-\gamma c(T)}e^{-\alpha T}x}\,,$$
	where $c(T):=\frac{\sigma^2}{2\alpha} (1-e^{-\alpha T}).$ 
\end{lemma} 
\begin{proof} 
	See Corollary 6.3.4.4 in \cite{jeanblanc2009mathematical}, where the proof is given for $\gamma<0$; the same proof holds for $\gamma<2\alpha/\sigma^2.$
\end{proof}

\section{$3/2$ model}
\label{app:3/2}

Consider the $3/2$ model 
$$dX_t=(b-aX_t)X_t\,dt+\sigma {X_t}^{3/2}\,dB_t\,,\;X_0=\xi$$
where $b,\sigma,\xi>0$, and 
$a>-\frac{\sigma^2}{2}.$ 
For
\begin{equation}
\label{eqn:app_3/2_q}
q>-\frac{1}{2\sigma^2}\Big(a+\frac{\sigma^2}{2}\Big)^2+\frac{\sigma^2}{8}\,,
\end{equation}
define $$\eta:=\frac{\sqrt{(a+\sigma^2/2)^2+2q\sigma^2}-(a+\sigma^2/2)}{\sigma^2} \,.$$ 
Then, it is easy to check that $\alpha:=a+\sigma^2\eta>0.$


We apply the Hansen--Scheinkman decomposition to estimate the expectation
$$p_T=\mathbb{E}^\mathbb{P}[e^{-q\int_0^TX_s\,ds}]\,.$$
The corresponding operator is
$$\mathcal{P}_Th(x)=\mathbb{E}_x^{{\mathbb{P}}}[e^{-q\int_0^T X_u\,du}h(X_T)]\,,$$
and it can be shown that 
$(\lambda,\phi(x)):=(b\eta, x^{-\eta})$
is an eigenpair.
The eigen-measure $\hat{\mathbb{P}}$ is defined on   $\mathcal{F}_T$ as
\begin{equation}
\frac{d\hat{\mathbb{P}}}{d\mathbb{P}}
=e^{-\frac{1}{2}\sigma^2\eta^2\int_0^T X_s\,ds-\sigma\eta \int_0^T\sqrt{X_s}\,dB_s}\,.
\end{equation} 
The process
$$\hat{B}_t=\sigma\eta\int_0^t\sqrt{X_s}\,ds+B_t\,,\;0\le t\le T$$
is a Brownian motion by the Girsanov theorem, and $X$ follows
\begin{equation}
\begin{aligned}
dX_t=(b-\alpha X_t)X_t\,dt+\sigma {X_t}^{3/2}\,d\hat{B}_t
\end{aligned}
\end{equation}
for
$\alpha=a+\sigma^2\eta.$

The expectation $p_T$ can be expressed as 
\begin{equation} \label{eqn:3/2_HS_decompo}
\begin{aligned}
p_T&=\mathbb{E}_\xi^\mathbb{P}[e^{-q\int_0^TX_s\,ds}]
=\mathbb{E}_\xi^{\hat{\mathbb{P}}}[X_T^{\eta}] \xi^{-\eta}e^{-\lambda T}=f(T,\xi)\xi^{-\eta}e^{-\lambda T},
\end{aligned}
\end{equation}
where \begin{equation}\label{eqn:3/2_reminder}
f(t,x)=\mathbb{E}_x^{\hat{\mathbb{P}}}[X_t^{\eta}] \,,\;0\le t\le T\,, x>0
\end{equation} 
is the remainder function.
Considering that $1/X$ is a CIR model and  $h$ has linear growth,
the function  
$f(T,\xi) $ converges to a constant 
as $T\to\infty.$

We study the large-time asymptotic behavior of the
remainder function $f$ and its sensitivity  with respect to the parameters $\alpha$ and $\eta.$

\begin{prop}\label{prop:app_3/2} 
	Suppose that $X$ follows 
	\begin{equation}
	\label{eqn:app_3/2}
	dX_t=(b-\alpha X_t)X_t\,dt+\sigma {X_t}^{3/2}\,d\hat{B}_t\,,\;X_0=\xi
	\end{equation} 
	for  $b,\alpha,\xi>0,$ $\sigma\neq0.$ Define 
	$$f(T,\xi)=\mathbb{E}_\xi^{\hat{\mathbb{P}}}[X_T^{\eta}]
	\,,\; T\ge0\,, $$
	for $\eta<\frac{2\alpha}{\sigma^2}+1.$ 
	Then,  
	\begin{equation}\label{eqn:quand_3/2_conv}
	f(T,\xi)\to \int x^\eta\,d\pi(x)
	\end{equation}
	as $T\to\infty$, where $\pi$ is the invariant distribution of $X.$ The partial derivatives $f_\eta(T,\xi)$ and
	$f_\alpha(T,\xi)$ are bounded in $T$ on $[0,\infty).$
\end{prop}

\begin{proof} 
It is easy to prove Eq.\eqref{eqn:quand_3/2_conv} by considering the density function of $X.$  
We prove that two partial derivatives $f_\eta(T,\xi)$ and
$f_\alpha(T,\xi)$
are bounded in $T$ on $[0,\infty).$
First, consider the partial derivative $f_\eta(T,\xi).$
Observe that
$$f_\eta(T,\xi)
	=\frac{\partial}{\partial\eta}\mathbb{E}_\xi^{\hat{\mathbb{P}}}[X_T^\eta]
	=\mathbb{E}_\xi^{\hat{\mathbb{P}}}\Big[\frac{\partial}{\partial\eta}X_T^\eta\Big]
	= \mathbb{E}_\xi^{\hat{\mathbb{P}}}[X_T^\eta \ln X_T]\,.$$
There are constants $c>0,$ $\gamma<\frac{2\alpha}{\sigma^2}+2$, and $m\in\mathbb{N}$ such that
$$x^\eta \ln x \leq cx^{\gamma}\,,\;x\ge1$$
and 
$$x^\eta |\ln x| \leq cx^{-m}\,,\;0<x<1\,.$$ 
Thus, by Lemma \ref{lem:3/2_MMG}, the partial derivative $f_\eta(T,\xi)$ is bounded in $T$ on $[0,\infty).$

Now we show that the partial derivative $f_\alpha(T,\xi)$
is bounded in $T$ on $[0,\infty).$ 
Define $H(x)=x^{\eta}$ for notational simplicity.
By Propositions \ref{prop:Malliavin_chain}, \ref{prop:IBP}, and \ref{prop:rho}, we have
\begin{equation} 
\begin{aligned}
f_\alpha(T,\xi)=\frac{\partial}{\partial\alpha}\mathbb{E}_\xi^{\hat{\mathbb{P}}}[H(X_T)]
&=-\frac{1}{\sigma}\mathbb{E}_\xi^{\hat{\mathbb{P}}}\Big[H(X_T)\int_0^T\sqrt{X_s}\,d\hat{B}_s\Big]\\
&=-\frac{1}{\sigma}\mathbb{E}_\xi^{\hat{\mathbb{P}}}\Big[ \int_0^TD_s(H(X_T))\sqrt{X_s}\,ds\Big] \\
&=-\frac{1}{\sigma}\mathbb{E}_\xi^{\hat{\mathbb{P}}}\Big[ \int_0^TH'(X_T) (D_sX_T)\sqrt{X_s}\,ds\Big]\\
&=-\mathbb{E}_\xi^{\hat{\mathbb{P}}}\Big[ \int_0^TH'(X_T)X_T^{3/2} e^{-\frac{b}{2}(T-s)} e^{-(\frac{\alpha}{2}+\frac{3\sigma^2}{8})\int_s^T X_u\,du}\sqrt{X_s}\,ds\Big] \,.
\end{aligned}
\end{equation} 
For the last equality, we used that the Malliavin derivative of $X_T$ is 
\begin{equation}
\begin{aligned}
D_sX_T 
=\sigma e^{-\frac{b}{2}(T-s)} e^{-(\frac{\alpha}{2}+\frac{3\sigma^2}{8})\int_s^T X_u\,du}X_T^{3/2}\,.
\end{aligned}
\end{equation}
From
\begin{equation}
\begin{aligned}
\sqrt{X_T}=\sqrt{X_s}e^{\frac{b}{2}(T-s)-(\frac{\alpha}{2}+\frac{\sigma^2}{4})\int_s^TX_u\,du+ \frac{\sigma}{2} \int_s^T X_u^{1/2}\,d\hat{B}_u}\,,
\end{aligned}
\end{equation}
we have
\begin{equation} \label{eqn:3/2_f_alpha}
\begin{aligned}
f_\alpha(T,\xi) 
&=-\mathbb{E}_\xi^{\hat{\mathbb{P}}}\Big[ \int_0^TH'(X_T)X_T  e^{-(\alpha+\frac{5\sigma^2}{8})\int_s^TX_u\,du+ \frac{\sigma}{2} \int_s^T X_u^{1/2}\,d\hat{B}_u} X_s\,ds\Big]\\
&=-\mathbb{E}_\xi^{\hat{\mathbb{P}}}\Big[ \int_0^T\mathbb{E}_\xi^{\hat{\mathbb{P}}}[H'(X_T)X_T  e^{-(\alpha+\frac{5\sigma^2}{8})\int_s^TX_u\,du+ \frac{\sigma}{2} \int_s^T X_u^{1/2}\,d\hat{B}_u}|X_s] X_s^2\,ds\Big]\\
&=-\mathbb{E}_\xi^{\hat{\mathbb{P}}}\Big[ \int_0^Tg(T-s,X_s)X_s^2\,ds\Big],
\end{aligned}
\end{equation} 
where $g$ is defined as
\begin{equation}
\begin{aligned}
g(t,x)
&=\mathbb{E}_x^{\hat{\mathbb{P}}}[H'(X_t)X_t  e^{-(\alpha+\frac{5\sigma^2}{8})\int_0^tX_u\,du+ \frac{\sigma}{2} \int_0^t X_u^{1/2}\,d\hat{B}_u}]\\
&=\mathbb{E}_\xi^{\hat{\mathbb{P}}}[H'(X_{t+s})X_{t+s}  e^{-(\alpha+\frac{5\sigma^2}{8})\int_s^{t+s}X_u\,du+ \frac{\sigma}{2} \int_s^{t+s} X_u^{1/2}\,d\hat{B}_u}|X_s=x]\,.
\end{aligned}
\end{equation} 
The second equality is from the Markov property.

We aim to estimate this function $g.$
Define a measure $\tilde{\mathbb{P}}$ on $\mathcal{F}_T$ as
$$\frac{d\tilde{\mathbb{P}}}{d\hat{\mathbb{P}}}=e^{-\frac{\sigma^2}{8}\int_0^TX_u\,du+ \frac{\sigma}{2} \int_0^T X_u^{1/2}\,d\hat{B}_u}\,.$$
It is easy to show that a local martingale $(e^{-\frac{\sigma^2}{8}\int_0^tX_u\,du+ \frac{\sigma}{2} \int_0^t X_u^{1/2}\,d\hat{B}_u})_{0\le t\le T}$ is a martingale.
The process $X$ satisfies 
$$dX_t=(b-(\alpha-\frac{1}{2}\sigma^2) X_t)X_t\,dt+\sigma {X_t}^{3/2}\,d\tilde{B}_t\,,$$
where $(\tilde{B}_t)_{0\le t\le T}$ is a $\tilde{\mathbb{P}}$-Brownian motion.
Note that since the mean-reversion speed $\alpha-\frac{1}{2}\sigma^2>-\frac{1}{2}\sigma^2,$ the process stays positive under the measure $\tilde{\mathbb{P}}.$ 
It follows that 
$$g(t,x)=\mathbb{E}_x^{\tilde{\mathbb{P}}}[H'(X_t)X_t  e^{-(\alpha+\frac{\sigma^2}{2})\int_0^tX_u\,du}]\,.$$
We apply the Hansen--Scheinkman decomposition here.
Consider the operator
$$h\mapsto\mathbb{E}_x^{\tilde{\mathbb{P}}}[h(X_t)
e^{-(\alpha+\frac{\sigma^2}{2})\int_0^tX_u\,du}]\,.$$
It can be shown that 
$(\tilde{\lambda},\tilde{\phi}(x)):=(b, x^{-1})$
is an eigenpair, and let $\overline{\mathbb{P}}$ be the corresponding eigen-measure. The $\overline{\mathbb{P}}$-dynamics of $X$ is 
$$dX_t=(b-(\alpha+\frac{1}{2}\sigma^2) X_t)X_t\,dt+\sigma {X_t}^{3/2}\,d\overline{B}_t\,,$$
where $(\overline{B}_t)_{0\le t\le T}$
is a $\overline{\mathbb{P}}$-Brownian motion.
Then,
$$g(t,x)=\mathbb{E}_x^{\tilde{\mathbb{P}}}[H'(X_t)X_t  e^{-(\alpha+\frac{\sigma^2}{2})\int_0^tX_u\,du}]=
\mathbb{E}_x^{\overline{\mathbb{P}}}[H'(X_t)X_t^2]e^{-bt}x^{-1}\,.$$ 
Since 
$|H'(x)|x^2=\eta x^{\eta+1}$
and $\eta+1\leq \frac{2\alpha}{\sigma^2}+3$ holds,
by Lemma \ref{lem:3/2_MMG} (with $\alpha$ replaced by $\alpha+\frac{1}{2}\sigma^2$), the expectation 
$\mathbb{E}_x^{\overline{\mathbb{P}}}[|H'(X_t)|X_t^2]$ is uniformly bounded in $(t,x)$ on $[0,\infty)\times (0,\infty).$
Thus,
$$|g(t,x)|\leq c'e^{-bt}x^{-1}$$
for some positive constant $c',$ which is independent of $t$ and $x.$
Eq.\eqref{eqn:3/2_f_alpha} gives
\begin{equation} 
\begin{aligned}
|f_\alpha(T,\xi)|  
&\leq \mathbb{E}_\xi^{\hat{\mathbb{P}}}\Big[ \int_0 ^T|g(T-s,X_s)| e^{-b(T-s)}X_s^2\,ds\Big]
\leq c'\int_0^T e^{-b(T-s)} \mathbb{E}_\xi^{\hat{\mathbb{P}}}[X_s^2]\,ds\,.
\end{aligned}
\end{equation} 
By Lemma \ref{lem:3/2_MMG}, the expectation $\mathbb{E}_\xi^{\hat{\mathbb{P}}}[X_s^2]$ is bounded in $s$ on $[0,\infty).$
This gives the desired result.
\end{proof}

\begin{lemma}\label{lem:3/2_MMG} 
	Let $X$ be a solution of
	$$dX_t=(b-\alpha X_t)X_t\,dt+\sigma {X_t}^{3/2}\,d\hat{B}_t\,,\;X_0=\xi$$	for $b,\sigma,\xi>0$, and $\alpha\ge-\frac{\sigma^2}{2}.$	Then, for $A<\frac{2\alpha}{\sigma^2}+2,$ 	\begin{equation}\label{eqn:3/2_moment}	\begin{aligned}	H(T,\xi):=\mathbb{E}_\xi(X_T^A)	&=\frac{\Gamma(\frac{2\alpha}{\sigma^2}+2-A)}{\Gamma(\frac{2\alpha}{\sigma^2}+2)}\Big(\frac{2b}{\sigma^2}\frac{1}{1-e^{-bT}}\Big)^{A}	F\Big(A,\frac{2\alpha}{\sigma^2}+2,-\frac{2b}{\sigma^2}\frac{1}{(e^{bT}-1)\xi}\Big)\,,
	\end{aligned}
	\end{equation}
	where $F$ is the confluent hypergeometric function.
	The function $H(T,\xi)$ converges to $$\frac{\Gamma(\frac{2\alpha}{\sigma^2}+2-A)}{\Gamma(\frac{2\alpha}{\sigma^2}+2)}\Big(\frac{2b}{\sigma^2}\Big)^{A}$$ as $T\to\infty.$
	Moreover, if $0<A<\frac{2\alpha}{\sigma^2}+2,$ then 	the function $H$ is uniformly bounded on the domain $[0,\infty)\times(0,\infty).$\end{lemma} 

\noindent See \cite[Lemma B.1]{park2019convergence} for the proof.

\section{Quadratic drift model}
\label{app:quad}

Assume that $X$ follows
\begin{equation}\label{eqn:garch_app}
dX_t= (b-aX_t^2)\,dt+ \sigma X_t\,dB_t\,,\;X_0=\xi
\end{equation} 
for $b,a,\sigma,\xi>0.$
This SDE has a unique strong solution and the solution stays positive by \cite[Proposition 2.1]{carr2019lognormal}.
For $q>-\frac{a^2}{2\sigma^2},$ we 
define
$$p_T:=\mathbb{E}_\xi^\mathbb{P}[e^{-q\int_0^TX_u^2\,du}]\,.$$
Consider the generator  
$$(\mathcal{L}\phi)(x):=\frac{1}{2}\sigma^2x^2\phi''(x)+(b-ax^2)\phi'(x)-qx^2\phi(x)\,.$$
It can be shown that 
$(\lambda,\phi):=(b\eta,e^{-\eta x})$
is an eigenpair, where
$$\alpha:=\sqrt{a^2+2q\sigma^2}\,,\;\eta:=\frac{\alpha-a}{\sigma^2}\,.$$
Let $\hat{\mathbb{P}}$ be the eigen-measure on $\mathcal{F}_T$
defined as
$$\frac{d\hat{\mathbb{P}}}{d\mathbb{P}}=e^{-\frac{1}{2}\sigma^2\eta^2\int_0^T X_s^2 \,ds-\sigma\eta\int_0^TX_s\,dB_s}\,.$$
Then, the $\hat{\mathbb{P}}$-dynamics of $X$
is
$$dX_t= (b-\alpha X_t^2)\,dt+ \sigma X_t\,d\hat{B}_t\,,\;X_0=\xi$$
for a $\hat{\mathbb{P}}$-Brownian motion $(\hat{B}_t)_{0\le t\le T}.$ It can be easily checked that the local martingale $(e^{-\frac{1}{2}\sigma^2\eta^2\int_0^t X_s^2 \,ds-\sigma\eta\int_0^tX_s\,dB_s})_{0\le t\le T}$ is a martingale.
It follows that
\begin{equation}
\label{eqn:quad_HS_decompo}
p_T:=\mathbb{E}_\xi^\mathbb{P}[e^{-q\int_0^TX_s^2\,ds}]=f(T,\xi)e^{-\eta\xi}e^{-\lambda  T},
\end{equation} 
where
$$f(t,x):=\mathbb{E}_x^{\hat{\mathbb{P}}}[e^{\eta X_t}]\,.$$
The invariant measure of $X$ under the measure $\hat{\mathbb{P}}$ is
$$d\pi(x)=\frac{1}{\sigma^2x^2}e^{-\frac{2b}{\sigma^2 }\frac{1}{x}-\frac{2\alpha}{\sigma^2}x}\,dx$$
up to positive constant multiples.
For further details on 
the invariant measure, readers may refer to
\cite[Proposition 3.2]{locherbach2013ergodicity}
or 
\cite[Lemma 20.19]{kallenberg2006foundations}.

\begin{prop}\label{prop:quad_model_remain}
	Suppose that $X$ follows 
	\begin{equation}
	\label{eqn:app_quad}
	dX_t= (b-\alpha X_t^2)\,dt+ \sigma X_t\,d\hat{B}_t\,,\;X_0=\xi
	\end{equation} 
	for  $b,\alpha,\xi>0,$ $\sigma\neq0.$ Define 
	$$f(T,\xi):=\mathbb{E}_\xi^{\hat{\mathbb{P}}}[e^{\eta X_T}]\,,$$
	for  $\eta< \frac{\alpha}{\sigma^2}.$ 
	Then, 
	\begin{equation}\label{eqn:quad_quad_conv}
	f(T,\xi)\to \int e^{\eta x}\,d\pi(x)
	\end{equation}
	as $T\to\infty$, where $\pi$ is the invariant distribution of $X.$ The partial derivatives $f_\eta(T,\xi)$ and
	$f_\alpha(T,\xi)$ are bounded in $T$ on $[0,\infty).$	
\end{prop}

\begin{proof}
	Since $e^{\eta x}\in L^2(m),$
	by Proposition \eqref{eqn:quad_model_L_2_con}, we have
	Eq.\eqref{eqn:quad_quad_conv}.
	By the same method in Eq.\eqref{eqn:f_eta_OU_2}, it follows that $f_\eta(T,\xi)$ is bounded in $T.$	
	We now show that the partial derivative $f_\alpha(T,\xi)$
	is bounded in $T$ on $[0,\infty).$ 
	Define $H(x)=e^{\eta x}$ for notational simplicity.
	We have
	\begin{equation} 
	\begin{aligned}
	f_\alpha(T,\xi)=\frac{\partial}{\partial\alpha}\mathbb{E}_\xi^{\hat{\mathbb{P}}}[H(X_T)]
	&=-\frac{1}{\sigma}\mathbb{E}_\xi^{\hat{\mathbb{P}}}\Big[H(X_T)\int_0^T X_s\,d\hat{B}_s\Big]\\
	&=-\frac{1}{\sigma}\mathbb{E}_\xi^{\hat{\mathbb{P}}}\Big[ \int_0^TD_s(H(X_T)) X_s\,ds\Big] \\
	&=-\frac{1}{\sigma}\mathbb{E}_\xi^{\hat{\mathbb{P}}}\Big[ \int_0^TH'(X_T) (D_sX_T) X_s\,ds\Big]  \,.
	\end{aligned}
	\end{equation} 
	By Proposition \ref{prop:quad_malliavin},
	the Malliavin derivative of $X_T$ is
	$$D_sX_T
	=\sigma X_se^{-2\alpha\int_s^TX_s\,ds-\frac{1}{2}\sigma^2(T-s)+\sigma (\hat{B}_T-\hat{B}_s)}=\sigma X_Te^{-\int_s^T (\frac{b}{X_s}+\alpha X_s)\,ds}\,,\;0\le s\le T\,.$$
	The last inequality is from 
	$X_T/X_s= e^{ \int_s^T(\frac{b}{X_s}-\alpha X_s)\,ds-\frac{1}{2}\sigma^2(T-s) +\sigma (\hat{B}_T-\hat{B}_s)}.$
	Thus,
	\begin{equation} 
	\begin{aligned}
	|f_\alpha(T,\xi)|
	&\leq \frac{1}{\sigma}\mathbb{E}_\xi^{\hat{\mathbb{P}}}\Big[ \int_0^T |H'(X_T)| (D_sX_T) X_s\,ds\Big]\\  
	&=\mathbb{E}_\xi^{\hat{\mathbb{P}}}\Big[ \int_0^T |H'(X_T)| X_Te^{-\int_s^T (\frac{b}{X_s}+\alpha X_s)\,ds}X_s\,ds\Big]\\  
	&\leq \mathbb{E}_\xi^{\hat{\mathbb{P}}}\Big[ \int_0^T |H'(X_T)| X_Te^{-2\sqrt{\alpha b}(T-s)}X_s\,ds\Big]\quad \Big(\because \frac{b}{X_s}+\alpha X_s\geq 2\sqrt{\alpha b}\Big)\,.
	\end{aligned}
	\end{equation} 	
	Choose constants $p$ and $q$ such that $1<p<\frac{\alpha}{\alpha-a}$ and $1/p+1/q=1.$
	Then, 
	\begin{equation} \label{eqn:f_alpha_estimate}
	\begin{aligned}
	|f_\alpha(T,\xi)|  
	&\leq  \int_0^T \mathbb{E}_\xi^{\hat{\mathbb{P}}}[|H'(X_T) X_T|^p]^{1/p}\,\mathbb{E}_\xi^{\hat{\mathbb{P}}}[X_s^q]^{1/q}e^{-2\sqrt{\alpha b}(T-s)}\,ds\Big] \,.
	\end{aligned}
	\end{equation}
	Observe that 
	$|H'(x)x|^p=\eta^px^pe^{p\eta x}\in L^2(m)$
	since $p\eta<\frac{\alpha}{\sigma^2}.$
	By Proposition \ref{eqn:quad_model_L_2_con}, the expectation
	$\mathbb{E}_\xi^{\hat{\mathbb{P}}}[|H'(X_T)X_T|^p]$ 
	converges to $\int |H(x)x|^p\,m(dx)$ as $T\to\infty.$ In particular, the expectation is bounded in $T$ on $[0,\infty).$
	Similarly,  the expectation 
	$\mathbb{E}_\xi^{\hat{\mathbb{P}}}[X_s^q]$ is also bounded in $s$ on $[0,\infty).$ 
	Since the right-hand side of Eq.\eqref{eqn:f_alpha_estimate} is bounded in $T$ on $[0,\infty),$ we obtain the desired result.
\end{proof}

\begin{prop}\label{eqn:quad_model_L_2_con}
	Let $g\in L^2(m).$ Then, 
	$$\mathbb{E}_\xi^{\hat{\mathbb{P}}}[g(X_T)]\to\int_{(0,\infty)} g(x)\,m(dx)$$ 
	as $T\to\infty.$
\end{prop}

\begin{proof}

	Consider the eigenvalue--eigenfunction problem 
	$\mathcal{L}\phi=-\lambda\phi$ of the second-order differential operator
	$$\mathcal{L}\phi:=\frac{1}{2}\sigma^2x^2\phi''+(b-\alpha x^2)\phi'$$
	densely defined on the space $L^2(m).$
	We first show that the spectral gap is positive.
	By \cite[Theorem 12 (ii)]{fulton2005automatic}, $\mathcal{L}$ is bounded below since $\phi=1\in L^2(m)$ is non-oscillatory. 
	By \cite[Theorem 14]{fulton2005automatic}, it suffices to show that the essential spectrum is empty.  
	We can write the equation $\mathcal{L}\phi=-\lambda\phi$ in the divergence form as
	\begin{equation}\label{eqn:div_form}
	\begin{aligned}
	-(p(x)\phi'(x))'=\lambda w(x)\phi(x), 
	\end{aligned}
	\end{equation} 
	where
	$$p(x)=e^{-\frac{2b}{\sigma^2 }\frac{1}{x}-\frac{2\alpha}{\sigma^2}x}\,,\;w(x)=\frac{2}{\sigma^2x^2}e^{-\frac{2b}{\sigma^2 }\frac{1}{x}-\frac{2\alpha}{\sigma^2}x}\,.$$
	Using the Liouville transformation (for example, \cite[Section 7]{everitt2005catalogue}),
	Eq.\eqref{eqn:div_form} becomes
	$$-Y'(X)+Q(X)Y(X)=\lambda Y(X)\,,$$
	where
	$$Q(X)=-(p(x)/w^3(x))^{1/4}(p(x)((p(x)w(x))^{-1/4})')'$$
	and $X=\frac{\sqrt{2}}{\sigma}\ln x.$
	By direct calculation,
	$Q(X)\simeq x=e^{\frac{\sigma}{\sqrt{2}}X}$ as $X\to\infty;$ more precisely,
	$\lim_{X\to\infty}\frac{Q(X)}{e^{\frac{\sigma}{\sqrt{2}}X}}$ exists and is a positive constant.
	In particular, $Q(X)\to\infty$ as $X\to\infty$, and this implies that the essential spectrum is empty from
	\cite[Corollary 4.2]{curgus2002discreteness}. 	
	Finally, since the spectral gap is positive, by \cite[Theorem 5.2 (iii)]{qin2016positive},
	we have
	$$f(T,\xi):=\mathbb{E}_\xi^{\hat{\mathbb{P}}}[g(X_T)]\to \int_{(0,\infty)} g(x)\,m(dx)$$
	for $g\in L^2(m).$ This completes the proof.
\end{proof}

\begin{prop}\label{prop:quad_malliavin}
	Let $(X_t)_{t\ge0}$ be the solution of 
	$$dX_t=(b-\alpha X_t^2)\,dt+\sigma X_t\,d\hat{B}_t\,.$$
	Then, for $T>0,$ the random variable $X_T$ is in $\mathbb{D}^{1,2}$, and the Malliavin derivative is
	$$D_tX_T=\sigma X_te^{-2\alpha\int_t^TX_s\,ds-\frac{1}{2}\sigma^2(T-t)+\sigma (\hat{B}_T-\hat{B}_t)}$$	for $0\le t\le T.$ 
\end{prop}

\begin{proof}
	Choose $\tilde{b},\tilde{a}>0$ such that $b-ax^2<\tilde{b}-\tilde{a}x$ for all $x>0.$ 
	Let $\tilde{X}$ be the solution of the SDE
	$$d\tilde{X}_t=(\tilde{b}-\tilde{a}\tilde{X}_t)\,dt+\sigma \tilde{X}\,d\hat{B}_t\,.$$
	This SDE has a unique strong solution since the coefficients are  Lipschitz, and the solution $\tilde{X}$ stays positive by \cite[Eq.(0.2)]{zhao2009inhomogeneous}.
	For $N\in\mathbb{N},$ let $\Lambda_N$ be a continuously differentiable function satisfying
	\begin{equation}	\Lambda_N(x)=\left\{
	\begin{aligned}
	&b-ax^2 &&\textnormal{if } x\le N \\
	&\tilde{b}-\tilde{a}x &&\textnormal{if } x\ge \frac{\tilde{b}-b+aN^2}{\tilde{a}}+1
	\end{aligned}\right.
	\end{equation}
	as well as $\Lambda_N'(x)\leq 0$ for all $x>0.$ Since $\Lambda_N$ is a Lipschitz function, the SDE
	\begin{equation}
	\label{eqn:quad_lambda_SDE}
	dX_t^{(N)}=\Lambda_N(X_t^{(N)})\,dt+\sigma X_t^{(N)}\,d\hat{B}_t
	\end{equation} 
	has a unique strong solution $(X_t^{(N)})_{t\ge 0}.$
	By the comparison theorem (for example, see \cite[Proposition 2.18 in Chapter 5]{karatzas1991brownian}), we know that
	$X_t\leq X_t^{(N)}\leq \tilde{X}_t$ a.s.

	Since the Malliavin derivative is a closed operator, it suffices to show that $X_T^{(N)}\to X_T$ in $L^2$ as $N\to\infty$ and $D_tX_T^{(N)} \to\sigma X_te^{-2\alpha\int_t^TX_s\,ds-\frac{1}{2}\sigma^2(T-t)+\sigma (\hat{B}_T-\hat{B}_t)}$ in $L^2$ as $N\to\infty.$
	We prove these in the following two stops.
	The first step is to show that  the random variable  $X_T^{(N)}\to X_T$ in $L^2$ as $N\to\infty.$ We use the dominated convergence theorem to prove this. For each $N\in\mathbb{N},$ define a stopping time 
	$\tau_N:=\inf \{t\ge0\,|\,X_t\ge N\}.$ 
	Then, it is clear that $(\tau_N)_{N\in\mathbb{N}}$ is  nondecreasing  and $\lim \tau_N=\infty$ a.s.
	Let $X^{\tau_N}$ denote the stopped process of $X$ at $\tau_N$. Then, we have 
	$$X_T^{(N)}=X_T^{\tau_N},\;T\le \tau_N$$
	from the definition of $\Lambda_N.$ 
	Letting $N\to\infty,$ it follows that
	$\lim_{N\to\infty}X_T^{(N)}=\lim_{N\to\infty}X_T^{\tau_N}=X_T$ a.s. Observe that $(X_T^{(N)}-X_T)^2\leq (|X_T^{(N)}|+|X_T|)^2 \leq 4\tilde{X}_T^2 $  and $\mathbb{E}[\tilde{X}_T^2]<\infty$ by \cite[Corollay 2.2]{zhao2009inhomogeneous}. The dominated convergence theorem implies that
	$\mathbb{E}[(X_T^{(N)}-X_T)^2]\to 0$ as $N\to\infty.$
	
	The second step is to show that 
	$$D_tX_T^{(N)}=\sigma X_t^{(N)}e^{\int_t^T\Lambda_N'(X_s^{(N)})\,ds-\frac{1}{2}\sigma^2(T-t)+\sigma (\hat{B}_T-\hat{B}_t)}\to \sigma X_te^{-2a\int_t^TX_s\,ds-\frac{1}{2}\sigma^2(T-t)+\sigma (\hat{B}_T-\hat{B}_t)}\;\textnormal{ in } L^2$$
	as $N\to\infty.$ 
	Since the coefficients of SDE \eqref{eqn:quad_lambda_SDE} are continuously differentiable with bounded derivatives, the solution 
	$X_t^{(N)}\in\mathbb{D}^{1,2}$ for each $t$ and  $$D_tX_T^{(N)}=\sigma X_t^{(N)}e^{\int_t^T\Lambda_N'(X_s^{(N)})\,ds-\frac{1}{2}\sigma^2(T-t)+\sigma (\hat{B}_T-\hat{B}_t)}$$ by \cite[Property P2]{fournie1999applications}. Observe that $X_t^{(N)}\to X_t$ a.s. for each $t$  and $\Lambda_N'(x)\to -2ax$ for all $x>0$ as $N\to\infty.$ Thus, 
	$$\sigma X_t^{(N)}e^{\int_t^T\Lambda_N'(X_s^{(N)})\,ds-\frac{1}{2}\sigma^2(T-t)+\sigma (\hat{B}_T-\hat{B}_t)}\to \sigma X_te^{-2a\int_t^TX_s\,ds-\frac{1}{2}\sigma^2(T-t)+\sigma (\hat{B}_T-\hat{B}_t)}\;\textnormal{ a.s.}$$
	We claim that this convergence also holds in $L^2.$
	By the dominated convergence theorem, it suffices to find a $L^2$-dominating function.
	Considering that $\Lambda_N'(x)\leq0$ for $x>0,$ we have  
	\begin{equation}
	\begin{aligned}
	\sigma X_t^{(N)}e^{\int_t^T\Lambda_N'(X_s^{(N)})\,ds-\frac{1}{2}\sigma^2(T-t)+\sigma (\hat{B}_T-\hat{B}_t)}
	&\leq 	 \sigma X_t^{(N)}e^{-\frac{1}{2}\sigma^2(T-t)+\sigma (\hat{B}_T-\hat{B}_t)}\\
	&\leq  \sigma  \tilde{X}_te^{-\frac{1}{2}\sigma^2(T-t)+\sigma (\hat{B}_T-\hat{B}_t)}\,,
	\end{aligned}
	\end{equation}
	and similarly, $\sigma X_te^{-2a\int_t^TX_s\,ds-\frac{1}{2}\sigma^2(T-t)+\sigma (\hat{B}_T-\hat{B}_t)} 
	\leq  \sigma  \tilde{X}_te^{-\frac{1}{2}\sigma^2(T-t)+\sigma (\hat{B}_T-\hat{B}_t)}\,.$
	Thus,	
	\begin{equation}
	\begin{aligned}
	&\quad\big|\sigma X_T^{(N)}e^{\int_t^T\Lambda_N'(X_s^{(N)})\,ds-\frac{1}{2}\sigma^2(T-t)+\sigma (\hat{B}_T-\hat{B}_t)}-\sigma X_Te^{-2a\int_t^TX_s\,ds-\frac{1}{2}\sigma^2(T-t)+\sigma (\hat{B}_T-\hat{B}_t)}\big|\\
	&\leq 2\sigma  \tilde{X}_te^{-\frac{1}{2}\sigma^2(T-t)+\sigma (\hat{B}_T-\hat{B}_t)}\,.
	\end{aligned}
	\end{equation}
	The random variable $2\sigma  \tilde{X}_te^{-\frac{1}{2}\sigma^2(T-t)+\sigma (\hat{B}_T-\hat{B}_t)}$ is an $L^2$-dominating function 
	since $$\mathbb{E}[(2\sigma  \tilde{X}_te^{-\frac{1}{2}\sigma^2(T-t)+\sigma (\hat{B}_T-\hat{B}_t)})^2]=4\sigma^2 \mathbb{E}[\tilde{X}_t^4]^{1/2}\mathbb{E}[e^{-2\sigma^2(T-t)+4\sigma (\hat{B}_T-\hat{B}_t)}]^{1/2}<\infty$$ by \cite[Proposition 2.1]{zhao2009inhomogeneous}. This completes the proof.
\end{proof}




\bibliographystyle{apa}

\bibliography{sensitivity}
\end{document}